\documentclass[12pt]{elsarticle}

\journal{\url{arxiv.org}}

\usepackage{amsmath, amsfonts, amssymb, amsthm}
\usepackage{enumitem}
\usepackage[section]{placeins}
\usepackage[table,svgnames]{xcolor}
\usepackage{todonotes}
\usepackage{xstring} %
\usepackage{float} %
\usepackage{placeins} %
\usepackage{multirow}
\usepackage{mathtools}
\usepackage{graphicx}
\usepackage{tabularx}
\usepackage{calc}

\usepackage[utf8]{inputenc}    %
\usepackage[T1]{fontenc}  

\usepackage[vlined,ruled,boxed,linesnumbered]{algorithm2e}
\SetKw{br}{break}

\newtheorem{Thm}{Theorem}
\newtheorem{Cor}[Thm]{Corollary}
\newtheorem{Lem}[Thm]{Lemma}

\newtheorem{Defn}[Thm]{Definition}

\newtheorem{Claim}[Thm]{Claim}

\newcommand{\continue}{\textbf{continue}}

\newcommand{\avec}{\mathbf{a}}
\newcommand{\bee}{\mathbf{b}}
\newcommand{\dvec}{\mathbf{d}}
\newcommand{\e}{\mathbf{e}}
\newcommand{\f}{f}

\newcommand{\p}{\mathbf{p}}
\newcommand{\ve}{\mathbf{v}}

\newcommand{\dict}{\mathcal{D}}
\newcommand{\localdict}{\mathcal{L}}

\newcommand{\tupleset}{\mathcal{T}}

\newcommand{\createdict}{\mathtt{create\_dictionary}}

\newcommand{\appendto}{\mathtt{append\_to}}
\newcommand{\getitems}{\mathtt{get\_items}}

\newcommand{\lst}{\mathtt{lst}}
\newcommand{\key}{\mathtt{key}}
\newcommand{\val}{\mathtt{value}}

\newcommand{\algProb}{3/4}
\newcommand{\termProb}{9/10}

\newcommand{\la}{\langle}
\newcommand{\ra}{\rangle}

\newcommand{\bigoh}{\mathcal{O}}

\renewcommand{\f}{\ensuremath{f}}

\newcommand{\fterm}{\ensuremath{f}^{(\ell)}}
\newcommand{\ftermarg}[1]{\ensuremath{f^{(#1)}}}
\newcommand{\F}{\ensuremath{F}}

\newcommand{\FF}{{\mathbb{F}}}

\newcommand{\softoh}{\widetilde{\mathcal{O}}}
\newcommand{\sprod}{\textstyle\prod}
\newcommand{\ssum}{\textstyle\sum}
\newcommand{\SLP}{\ensuremath{\mathcal{S}_F}}
\newcommand{\SLPnoF}{\ensuremath{\mathcal{S}}}

\newcommand{\zero}{\mathbf{0}}
\newcommand{\z}{z}

\DeclareMathAlphabet\mathbfcal{OMS}{cmsy}{b}{n}
\newcommand{\Z}{\mathcal{Z}}
\newcommand{\ZZ}{\mathbb{Z}}
\newcommand{\I}{\mathcal{I}}
\newcommand{\ring}{\mathsf{R}}

\newcommand{\mdef}{\max(6, 2\lceil \log D \rceil, \lceil \tfrac{25}{8}\ln(4T) 
\rceil)}
\newcommand{\lambdadef}{\max(21, \tfrac{100}{3}(T-1)\ln D, 80n, 
\tfrac{10}{3}m\ln m)}
\newcommand{\sdef}{\lceil\log 40 + 2\log m + 2\log n + 2\log T\rceil}
\newcommand{\udef}{\lceil\log_q(2nD+1)\rceil}

\newcommand*{\MyDef}{\mathrm{def}}
\newcommand*{\eqdefU}{\ensuremath{\mathop{\overset{\MyDef}{=}}}}%
\newcommand*{\eqdef}{\,\mathop{\overset{\MyDef}{\resizebox{\widthof{\eqdefU}}{\heightof{=}}{=}}}\,}

\definecolor{darkgreen}{rgb}{0,.35,0}
\definecolor{darkblue}{rgb}{0,0,.35}
\definecolor{darkred}{rgb}{.35,0,0}
\usepackage[pdfauthor={Andrew Arnold, Mark Giesbrecht, and Daniel S. Roche}, 
pdftitle={Faster Sparse Multivariate Polynomial Interpolation of 
Straight-Line Programs},
colorlinks,linkcolor=DarkBlue,
bookmarksnumbered]{hyperref}
            
\begin{document}
\title{Faster Sparse Multivariate Polynomial Interpolation of 
Straight-Line Programs}

\author[uw]{Andrew Arnold}
\ead{a4arnold@uwaterloo.ca}
\ead[url]{www.AndrewArnold.ca}

\author[uw]{Mark Giesbrecht}
\ead{mwg@uwaterloo.ca}
\ead[url]{www.uwaterloo.ca/~mwg}

\author[usna]{Daniel S.\ Roche}
\ead{roche@usna.edu}
\ead[url]{www.usna.edu/cs/roche}

\address[uw]{Cheriton School of Computer Science,
  University of Waterloo,
  Waterloo, Ontario, Canada}

\address[usna]{Computer Science Department,
  United States Naval Academy,
  Annapolis, Maryland, USA}

\begin{abstract}

Given a straight-line program whose output is a polynomial function of
the inputs, we present a new algorithm to compute a concise
representation of that unknown function.
Our algorithm can handle any case where the unknown function
is a multivariate polynomial, with coefficients in an
arbitrary finite field, and with a reasonable number of nonzero
terms but possibly very large degree. It is competitive with
previously known sparse interpolation algorithms that work over
an arbitrary finite field, and provides an
improvement when there are a large number of variables.

\end{abstract}

\maketitle

\section{Introduction}

We consider the problem of interpolating  a sparse multivariate polynomial
$F$ over $\FF_q$, the finite field of size $q$:  
\begin{gather}\label{eqn:fdef}
\F \quad=\quad \sum_{\ell=1}^t c_\ell \z_1^{e_{\ell 1}} \z_2^{e_{\ell 2}} \cdots 
  \z_n^{e_{\ell n}}
\quad\in\quad \FF_q[\z_1,\ldots,\z_n].
\end{gather}
We suppose $\F$ is given by a {\em Straight-Line Program} (SLP), a
list of simple instructions performing operations $+$, $-$ and
$\times$ on inputs and previously computed values, which evaluates the
polynomial at any point.  We further suppose we are given bounds
$D > \max_j\deg_{\z_j}(\F)$ and $T \geq t$.  It is expected that the
bound $T$ tells us that $\F$ is {\em sparse}, i.e., that $T \ll D^n$,
the maximum number of terms.  The goal of our interpolation algorithm
is to obtain the $t$ nonzero coefficients $c_\ell \in \FF_q$
and corresponding exponents
$\e_\ell=(e_{\ell_1},\ldots,e_{\ell_n}) \in \ZZ^n$ 
of $\F$.  Our contribution is as follows.

\begin{Thm}\label{Thm:Main}
  Let $\F\in\FF_q[z_1,\ldots,z_n]$, and suppose we are given a
  division-free straight-line program $\SLP$ of length $L$ which
  evaluates $\F$, an upper bound $D\geq \max_j\deg_{\z_j}(\F)$, and an
  upper bound $T$ on the number of nonzero terms $t$ of $\F$. There
  exists a probabilistic algorithm which interpolates $\F$ with
  probability at least $3/4$.  The algorithm requires
  \begin{equation*}
\softoh\left( Ln(T\log D + n)(\log D + \log q)\log D
+ n^{\omega-1}T\log D + n^\omega\log D \right).
  \end{equation*}
  bit operations.\footnote{For two functions $\phi, \psi$,
  we say $\phi \in \softoh(\psi)$ if and 
  only if $\phi \in \bigoh( \psi \log^c \psi)$ for some constant $c \geq 0$.}
  \footnote{The constant $\omega < 2.38$ is the exponent of matrix multiplication, meaning
  that the product of two $n\times n$ matrices can be computed in
  $O(n^\omega)$ field operations.}
\end{Thm}
This probability may be increased to $1-\epsilon$ using standard 
techniques, with cost increased by a factor 
$\bigoh(\log(\epsilon^{-1}))$.

The rest of this introductory section puts our work in context and defines
the notation and problem definitions for the rest of the paper. The reader
who is already familiar with the area may wish to glance at our list of
notation in \ref{app:notation}, then skip to
Section~\ref{sec:overview}, where we give a high-level overview of the algorithm
referred to by Theorem 1 and work out a small illustrative example in
full detail. The end of Section~\ref{sec:overview} provides an
outline for the remainder of the paper.

\subsection{Background and related work}

Polynomial interpolation is a fundamental problem of computational
mathematics that dates back centuries to the classic work of
Newton, Waring, and Lagrange. In such settings, given a list of
$(n+1)$-dimensional points and some degree bounds, the coefficients
of the unique $n$-variate polynomial interpolating those points
is produced.

If the number of nonzero coefficients is relatively small, the unknown
function can be treated as an exponential sum, and the task becomes
that of finding the exponents and coefficients of only the nonzero terms.
This is the \emph{sparse interpolation} problem, and it differs crucially
from other interpolation problems not only in the representation of the
output, but also that of the input. Every efficient sparse interpolation 
algorithm of which we are aware requires some control over where the unknown
function is sampled, and typically takes as input some procedure or
black box that can evaluate the unknown sparse polynomial at any chosen
point.

The sparse interpolation problem has received considerable interest over
fields of characteristic zero. The classical Prony's method for
exponential sums from 1795 (which can be regarded as the genesis of sparse
interpolation) was later applied to sparse interpolation over the
integers \citep{BT88,Kal10a} and approximate complex numbers
\citep{GLL09,KLY11}.
Compressive sensing is a different approach for
approximate sparse interpolation which has the advantage of allowing
the evaluation points to be chosen at random from a certain distribution
\citep{CRT06,Don06}.  Sparse Fourier and Hadamard-Walsh transforms 
allow for the interpolation of a sparse, complex-valued polynomial given 
by its discrete Fourier transform, and can find reasonable sparse 
approximations to non-sparse polynomials \cite{HasIndKatPri12, 
KusMan93}.

As in the rest of this paper, define $n$, $T$, and $D$ to be
(respectively) the number of variables and known bounds on the number
of terms and degree of the unknown polynomial.  An
information-theoretic lower bound on the complexity of univariate
sparse interpolation is $\Omega(T(n\log D+\log q))$, the number of bits
used to encode $F$ in \eqref{eqn:fdef}.  This bound is (nearly) met by
Prony's \cite{Pro95} algorithm (as adapted to the
polynomial setting; see \cite{Kal10a}), which requires $O(T\log q+TL)$
bit operations when $n=1$ and under the implicit assumption that $q>D$.  Much of the
complexity of the sparse interpolation problem appears to arise from
the requirement to accommodate \emph{any} finite field. Prony's
algorithm is dominated by the cost of discrete logarithms in $\FF_q$,
for which no polynomial time algorithm is known in general.  When
there is a choice of fields (say, as might naturally arise in a
modular scheme for interpolating integer or rational polynomials) more
efficient interpolation methods have been developed.  Kaltofen
\cite{Kal88:frag} demonstrates a method for sparse interpolation over
$\FF_p$ for primes $p$ such that $p-1$ is smooth; see
\citep{Kal10:pasco} for further exposition.  In our notation, this
algorithm would require $\softoh(LnT\log D + n^2 T\log^2 D)$ bit 
operations.

\citet{Zip90,HR99,JM10} developed multivariate sparse interpolation
algorithms that work over arbitrary finite fields, but whose running
time is polynomial in $D$, the degree.  The running time of these
methods is expressed in our terminology below for comparison, but in
fairness we should point out that \citep{JM10} is more general than
our algorithm as it relies exclusively on classical polynomial
arithmetic and requires only black-box access to the unknown function.

In extreme cases, the size of
the sparse representation may not be polynomial in the degree $D$, but rather in
$\log D$, as this sparse representation stores only the
coefficients and exponents of nonzero terms. In such cases, black-box
access to the unknown function may not even be sufficient, since the
degree could be much larger than the number of field elements.
Typically, algorithms in this setting take as input a straight-line
program, which allows for evaluations in a field extension or modulo
an ideal.

The earliest algorithm for sparse interpolation of a straight-line
program over an arbitrary finite field
whose cost is polynomial in $n$, $T$, and $\log D$
--- that is, the sparse representation
size --- was presented by \citet{GarSch09}. A series of results
\cite{ArnGieRoc13,AGR14,AR14,GieRoc11} has made use of different
randomizations to improve the complexity of this approach, in particular
reducing the dependence on the sparsity $T$ to quasi-linear.

The aforementioned algorithms for sparse interpolation of straight-line
programs are essentially univariate algorithms, but can easily be
extended to handle multivariate polynomials by use of the well-known
Kronecker substitution. A separate paper from two co-authors at ISSAC
2014 \citep{AR14} presented a new randomization that achieves similar
aims as the Kronecker substitution but with decreased degrees for sparse
polynomials. This technique is sufficiently general that it can be
combined with a wide variety of univariate interpolation algorithms to
achieve faster multivariate interpolation.

\begin{table}[tbh]
\begin{center}
\begin{tabular}{l|l}
\multicolumn{1}{c|}{\textbf{Algorithm}} & 
\textbf{Soft-O cost of SLP interpolation} \\
\hline
\citet{JM10} & $T^2 \log q (n+D) + LnT\log q$ \\
\citet{GarSch09} & $Ln^2 T^4 \log^2 D \log q + n^4T^4\log^3 D \log 
q$ \\
\citep{AGR14} with \citep{AR14} &
  $L n T (\log D + \log q)\log^2 D + n^\omega T$ \\[6 pt]
This paper (Thm 1) &
$Ln(T\log D + n)(\log D + \log q)\log D$ \\
& $+ n^{\omega-1}T(\log D) + n^\omega\log D$
\end{tabular}
\end{center}
\caption{Complexity of multivariate sparse interpolation algorithms%
\label{table:algs}}
\end{table}

Table~\ref{table:algs} summarizes the complexity of the best known
algorithms for multivariate sparse interpolation of straight-line
programs over an arbitrary finite field. Our new algorithm in this paper
uses many of the ideas in our ISSAC 2014 work \citep{AR14,AGR14}, but
synthesizes them in a novel way and reduces the amount of linear algebra
required. The cost is similar to the combination of our ISSAC 2014
results, but will be faster when the number of variables $n$ is
sufficiently large that the cost of computing $O(T)$ matrix inverses
dominates the complexity of the previous approach.

As a concrete example that may simplify the results of
Table~\ref{table:algs} and highlight the improvements here, suppose
the size of the SLP is equal to the number of nonzero terms ($L=T$),
which in turn twice the number of variables ($T=2n$), and the 
degree polynomials in the number of terms ($D < T^{\bigoh(1)}$). 
Then the cost of 
the algorithm in 
\citep{JM10} is at least $\softoh(n^4)$, the algorithm obtained from
our ISSAC 2014 results has cost $\softoh(n^{\omega + 1})$, and the new
algorithm presented in this paper further reduces the complexity in 
this case to $\softoh(n^3)$.

\subsection{Conventions and Notation}
\label{subsec:notation}

The technical nature of our results unfortunately necessitates a fair bit of
notation. In an effort to unburden the reader, we define the most
important facets of our notation here, and provide a reference
table in \ref{app:notation} that contains the names and conventions
used throughout the paper.

We will write vectors in boldface and vector entries in standard
typeface (with subscripts).  Column vectors will be written as
comma-separated tuples.  That is, we will write a length-$n$ column
vector as $\ve=(v_1, \dots, v_n)$, whereas $\ve_1, \dots, \ve_\ell$
denotes a list of $\ell$ vectors where each
$\ve_i = (v_{i1}, \dots, v_{in})$.  

A slight exception to this is that
we will let $\Z$ be the vector of indeterminates $(z_1,\ldots,z_n)$,
and hence $\FF_q[z_1,\ldots,z_n]$ is written $\FF_q[\Z]$.
For any exponent vector $\e=(e_1,\ldots,e_n)\in\ZZ_{\geq 0}^n$ we write $\Z^\e$
for the term $z_1^{e_1}\cdots z_n^{e_n}$. Similarly for any vector
$\avec=(a_1,\ldots,a_n)$, we write $\avec^\e$ for the product
$a_1^{e_1}\cdots a_n^{e_n}$.
For a single indeterminate $x$,
the notation $x^\e$ means $x^{e_1}\cdots x^{e_n}=x^{e_1+\cdots+e_n}$.

We assume a polynomial image $G(x) \bmod H(x)$ is represented in its
reduced form.  That is, we store the image as $R(x)$, where $G=HQ+R$,
$\deg_x R < \deg_x H$.  Similarly, an integer $s$ reduced modulo
$t \in \ZZ_{>0}$ is represented by $r$ such that $s=tq+r$, where
$q,r\in\ZZ$ and $0\le r<t$.

We use ``soft-oh'' notation for our cost analysis.  For an integer
constant $k$ and two functions $\phi, 
\psi : \mathbb{R}_{>0}^k \rightarrow \mathbb{R}_{>0}$, we say $\phi \in 
\softoh(\psi)$ if and only if $\phi \in \bigoh( \psi \log^c \psi)$ for some 
constant $c>0$.  Our algorithm will depend on asymptotically fast matrix 
multiplication.  By \cite{LeG14}, the cost of multiplying two $n \times n$ 
matrices over a field $\mathsf{F}$ entails $\bigoh( n^\omega)$ arithmetic 
operations in $\mathsf{F}$, where we can take $\omega=2.3728639$.

For $q=p^v$, $p$ prime, we suppose that elements of $\FF_q$ is 
represented as $\ZZ_p[x]/\la \Phi(x) \ra$, where $\Phi$ is a 
degree-$v$ irreducible polynomial over $\ZZ_p$.  We further 
identify $\FF_{q^u}$ with $\FF_q[y]/\la \Psi(y) \ra$, for $\Psi$ 
an irreducible polynomial over $\FF_q$.  Under such
representation, arithmetic operations in $\FF_{q^u}$ may be 
computed in $\softoh(u)$ arithmetic operations in $\FF_q$, each of which 
in turn may be computed in $\softoh( v )$ arithmetic operations in 
$\ZZ_p$ \cite{CanKal91}\footnote{See, e.g., Thm. 9.6 of 
\cite{MCA} for a cost analysis for division 
operations.}.  Multiplication in $\ZZ_p$ 
may be performed 
with $\softoh(\log p)$ bit operations (e.g., refer to \cite{Fur09} and Thm. 
9.8 of \cite{MCA}).  It follows that one can perform an 
arithmetic operation in $\FF_{q^u}$ in $\softoh( u \log q)$ bit 
operations.

Our algorithm will require randomness.  We assume that we may 
obtain a random bit with bit-cost $\bigoh(1)$, and that
the cost of choosing $x$ from a set $S$ uniformly at 
random admits a bit-cost of $\bigoh( \log|S|)$. The amount of randomness
required for our algorithm is stated in Lemma~\ref{lem:random}.

The polynomial 
$\F$ we are interpolating will be assumed to possess the following 
description and features throughout the remainder of this article:

\begin{align*}
  \F &= \sum_{\ell=1}^t c_\ell \Z^{\e_\ell} \in \FF_q[\Z], \\
  D & \geq \max_{j \in [n]}\deg_{\z_j}(\F),\\
  T & \geq t,\\
  \fterm &= c_\ell\Z^{\e_\ell}, \\
  \I_{i} & = \langle z_1^{p_i}-1, \ldots, z_n^{p_i}-1\rangle, \\
  \F_i &= \F \bmod \I_{p_i} \in \FF_q[\Z]/\I_{p_i}, \\
  \F_{ij} &= \F(x^{\ve_{ij}}) \bmod (x^{p_i}-1) \in 
  \FF_q[x]/\la x^{p_i}-1 \ra,\\
  \F_{ijk} &= \F(\avec_k x^{\ve_{ij}}) \bmod (x^{p_i}-1) \in 
             \FF_{q^u}[x]/\la x^{p_i}-1 \ra.
\end{align*}

\subsection{Straight-Line Programs}

A Straight-Line Program (SLP) is a branchless sequence of 
arithmetic instructions that may represent a rational function.

\begin{Defn}
A {\em Straight-Line Program} (SLP) 
over a ring $\ring$ with inputs $\z_1, \dots, \z_n$, is a sequence 
of arithmetic instructions 
$\SLPnoF =(\SLPnoF_1, 
\dots, 
\SLPnoF_L)$ of the form $\SLPnoF_i = (\beta_i \leftarrow \alpha 
\star \alpha')$, where $\star \in \{ +, -, \cdot, \div\}$ and 
$$
\alpha, \alpha' \in \{z_1, \dots, z_n\} \cup \{\beta_j \mid j < 
i\} \cup 
\ring.$$
We say $\beta$ 
(if well-defined) is the output for a choice of inputs $\z_1, 
\dots, \z_n$.
\end{Defn}
We say $\SLPnoF$ is {\em $n$-variate} if it accepts $n$ inputs.  
We let $L$ denote the length, or number of instructions, of an SLP.
 We henceforth restrict our attention to {\em division-free} SLPs 
 over $\FF_q$, so as to avoid possible division by zero.

Since an SLP gives us a list of arithmetic instructions, 
we may choose inputs $z_i$ from $\ring$, homomorphic images of 
$\ring$, or ring extensions thereof.  For instance, we may treat 
$\z_i$ as indeterminates, in which case, the resulting outputs 
$\beta_i$ each are polynomials from $\FF_{q}[\Z]$.  We say 
$\SLPnoF$ {\em computes} $\F \in \FF_{q}[\Z]$ if it outputs 
$\beta_L = \F$ given indeterminate inputs $\z_1, \dots, \z_n$.  We 
write $\SLP$ to denote 
an SLP that computes $\F$.  Every division-free SLP over $\FF_q$ 
computes some $\F \in \FF_q[\Z]$.

The aim of {\em sparse interpolation}, given an SLP $\SLP$, is 
to construct a {\em sparse representation} of $\F$: a list of 
nonzero terms of distinct degree comprising $\F$.  For instance, 
$\F=5z_1^6 + 7z_1^2 z_2^3$ admits a sparse representation 
$((5,(6,0)),(7,(2,3))$.

One could naively interpolate $\SLP$ by treating inputs $\z_1, 
\dots, \z_n$ as indeterminates and performing each arithmetic 
instruction as an arithmetic operation in $\FF_q[\Z]$.  A 
caveat of such an approach is that intermediate results $\beta_i$ 
may have arbitrarily many terms or degree with respect to that of 
$\F$.  Such an approach, in the worst case, has cost exponential 
in $L$.

Instead, we will use $\SLP$ to compute homomorphic images of 
$\F$.  We construct images of the form $\F(a_1 x^{v_1}, \dots, 
a_n x^{v_n}) \bmod (x^p-1)$, for appropriate choices of $p \in 
\ZZ_{>0}$.  Computing images of this form is preferable because 
it bounds the cost of executing each arithmetic instruction of 
$\SLP$.  
Namely, we perform arithmetic in $\FF_q[x]/\la x^p-1\ra$.  Using 
FFT-based techniques, one can perform arithmetic operations 
in $\ring[x]/\la x^p-1 \ra$, $\ring$ a ring, in $\softoh( p )$ 
arithmetic operations in $\ring$ \cite{CanKal91}.  Sometimes we 
will have to choose $a_i$ from a ring extension $\FF_{q^u}$, in 
which case we construct an image of $\F \in \FF_{q^u}[x]/ \la 
x^p-1 \ra$.  Per the previous discussion of finite field 
arithmetic in section \ref{subsec:notation}, this gives the 
following cost, which we state as a claim.

\begin{Claim}\label{claim:probe}
Given a length-$L$ division-free SLP $\SLP$ computing $\F \in 
\FF_q[\Z]$, one can construct an image of $\F$ in 
$\FF_{q^u}[x]/\la x^p-1 \ra$ in $\softoh( Lpu )$ operations in 
$\FF_q$, or $\softoh( Lpu \log q)$ bit operations.
\end{Claim}

\section{Overview\label{sec:overview}}

Our algorithm works by using three types of randomized 
homomorphisms to
compress the size of the problem, while still maintaining sufficient
information for reconstruction.  The randomization ensures that there
is not too much ``information collapse'' and we can correlate the
different homomorphic images with terms of the sparse polynomial we are
attempting to recover.  
A complete description of the algorithm is given
later in Procedure~\ref{proc:si},  along with a complete analysis, but we present the main ideas here.

\subsection{A high-level description of the algorithm}

In our algorithm we construct homomorphic images of $\F$ in order to
obtain information about the terms of $\F$.  The homomorphisms
compress the sparse polynomial so that it is amenable to efficient
interpolation, while still maintaining the fundmental structure.  A
final stage of the algorithm reconstructs the sought after sparse
polynomial from its homomorphic images.  We employ three distinct
types of homomorphism.  The first reduces the degree of each variable,
the second transforms the problem into a univariate problem while the
final one ensures that all the terms are distinct (and so can be
identified during reconstruction).

For the first type of homomorphism, we truncate, or wrap, each of the
variables $z_j$ modulo a prime $p_i$.  In fact, we construct this same
homomorphism $m$ times, with a collection of random primes
$p_1,\ldots,p_m$, selected once for the entire computation.
Specifically, let $\F_i$ be the polynomial obtained by replacing
$z_j\in\FF_q[\Z]$ by $z_j \bmod (z_j^{p_i}-1)$ for all $j\in [n]$.
Letting $\I_{p_i}$ be the ideal
$\langle z_1^{p_i}-1, \ldots,z_n^{p_i}-1\rangle$, we define
\begin{align*}
\Phi_{i}:~~~~~~~~ \FF_q[\Z] \to & \FF_q[\Z]/\I_i\\
z_1^{e_1}\cdots z_n^{e_n} \mapsto & z_1^{e_1\bmod p_i}\cdots
                                    z_n^{e_n\bmod p_i},
\end{align*}
so that $\F_i = \Phi_i(\F)$.
We say two terms $c\Z^{\e}$ and $c'\Z^{\e'}$ (alternatively, their
exponents) {\em collide} under this homomorphism if
$\Phi_i(\z^{\e})=\Phi(\z^{\e'})$.  This happens precisely when
$\e\equiv\e'(\bmod p_i)$, and we say that terms $c\Z^{\e}$ and
$c\Z^{\e'}$ collide in $\F_i$ in this case.  We define a {\em collision} to be a set
of size at least two comprising all the terms of $\F$ whose
exponents all agree under $\Phi_i$.

The second homomorphism builds on the first type, by not only truncating
the degrees but also reducing the number of variables to a new one $x$.  For each prime
$p_i$, we choose $n$ random vectors
$\ve_{i1},\ldots,\ve_{in}\in\ZZ_p^n$.  The homomorphism $\Psi_{ij}$ is
then defined by
\begin{align*}
\Psi_{ij}: \FF_q[\Z] & \to {\FF_q[x]}/{\langle x^p-1\rangle} \\
       z_k  & \mapsto x^{v_{ijk}\bmod p_i}.
\end{align*}
This further reduces the size of images we need to compute by an
exponential factor in $n$.  

Consider the term $f=c\z^\e=cz_1^{e_1}\cdots z_n^{e_n}$ of $\F$, which
maps under $\Psi_{ij}$ to
\[
\Psi_{ij}(f)
=
c x^{(e_1v_{ij1} + \cdots + e_{ijn}v_{ijn}) \bmod p_i}
=
c x^{\ve_{ij}\cdot\e} \bmod (x^{p_i}-1),
\]
and suppose its image has degree $d_{ij}<p_i$.
Then, combining the results from $n$ such homomorphisms,
$(e_1,\ldots,e_n) = \e \bmod {p_i}$ is the solution to
the linear system
\begin{equation}
  \label{eqn:linsys}
  \underbrace{
  \begin{bmatrix}
    v_{i11} & \cdots & v_{i1n} \\
   \vdots &           & \\
   v_{in1} & \cdots & v_{inn}
  \end{bmatrix}}_{V_i}
  \underbrace{
  \begin{bmatrix}
    e_1\\ \vdots \\ e_n
  \end{bmatrix}}_{\e}
  \equiv
  \underbrace{
  \begin{bmatrix}
    d_{i1} \\ \vdots \\ d_{in}
  \end{bmatrix}}_{\dvec_i}
  \bmod\ p_i.
\end{equation}

The basic idea of our algorithm is then to first make a selection of
random primes $p_1,\ldots,p_m$, and for each prime choose a random
matrix $V_{i}\in\ZZ_{p_i}^{n\times n}$ and compute its inverse.  
Now for each term $\f=c \Z^\e $ of $\F$, and for each prime $p_i$,
we obtain a vector of degrees
$\dvec_i$.  Using linear algebra as above, we can determine the
degree vector $(e_{i 1},\ldots,e_{i n})=\e \bmod p_i$.

In order for us to identify such a vector $\dvec_i$, we require that
\emph{(i)} $\f$ is not in a collision in any of the images
$\F_{ij}$, for any $j \in [n]$, and \emph{(ii)} 
all the other terms in images $\F_{ij}$ have a distinct coefficient from
that of $\f$.  The
first criterion is probabilistically ensured by our choice of images
$\F_{ij}$.  In order to obtain the latter, we need to employ a third
type of homomorphism which randomizes coefficients of terms of $\F$, as
well as coefficients arising from collisions of terms.  
This technique is called \emph{diversification} after its introduction
in \citep{GieRoc11}.

Specifically, {\em before} applying the aforementioned mappings to 
$\F$, we first choose a small number $s$ of random vectors
$\avec_1,\ldots,\avec_s$ over a field extension $\FF_{q^u}$, each
of which defines a mapping
\[
  (\z_1,\ldots,\z_n) \mapsto (a_{k1}\z_1,\ldots,a_{kn}\z_n).
\]
We then apply the latter homomorphisms and construct
\[
  \F_{ijk} \eqdef \sum_{\ell=1}^t c_\ell \avec_k^{\e_\ell}\Z^{\e_\ell\ve_j 
\bmod p_i},
\]
for $(i,j,k) \in [m,n,s]$.  

Now observe that the collisions only occur according to the choices of
primes $p_i$ and exponent vectors $\ve_j$; the different choices for
$\avec_k$ affect the coefficients in each image but not the exponents.
So for each pair $(i,j)\in [m,n]$, we have a sequence of degrees of nonzero
terms that appear in any $\F_{ijk}$ for $k\in[s]$. These are exactly the
degrees of nonzero terms in the polynomials determined by the
second homomorphism above, $\F_{ij} = \Psi_{ij}(\F)$.

In order to recover the complete multivariate exponents of the terms from these
images, we need to correlate these degrees in different images $\F_{ij}$ and
$\F_{i'j'}$, according to which term in $\F$ they correspond to.
This is where the third homomorphism (diversification) is used:
the random choice of $\avec_k$'s guarantees that, with high probability,
if the degree-$d$ term of $\F_{ij}$ does \emph{not} correspond to the degree-$d'$
term of $\F_{i'j'}$, there exists at least one $\avec_k$ such that
the degree-$d$ term of $\F_{ijk}$ has a different coefficient than the
degree-$d'$ term of $\F_{i'j'k}$. Hence the degrees in different images
$\F_{ij}$ can be grouped according to their coefficients in \emph{all}
of the images $\F_{ijk}$.

In our algorithm, this grouping is facilitated by a dictionary, for each
prime $p$, mapping coefficient vectors that appear in images 
$\F_{ijk}$ to their degrees. Whenever the same coefficient vector appears
for every vector $\ve_j, j\in[n]$, there is enough information to set up a linear
system as in \eqref{eqn:linsys} and recover that term's exponents modulo $p$.

For each term $f=c\Z^\e$ of $\F$,
the probability of actually obtaining $c$ and $\e \bmod p_i$
for some $i\in[m]$ will be shown to exceed $\termProb$.
By choosing sufficiently many primes $p_i$,
we can then guarantee that we will have enough 
information to construct all the terms of $\f$, and to detect all 
images of collisions, with probability $\algProb$.

\subsection{An Illustrative Example}

We consider an example which explains how we could use the 
suggested homomorphic images in order to construct terms of $\F$, 
as well as obstacles to our approach.  Suppose we are given an SLP 
that computes
$$
\F = \underbrace{z_1 z_2}_{\ftermarg{1}} + \underbrace{z_1^6 
z_2^6}_{\ftermarg{2}} + \underbrace{2z^4 z^{10}}_{\ftermarg{3}} + 
\underbrace{4z_1^3 z_2^{20}}_{\ftermarg{4}} \in 
\FF_{13}[z_1,z_2],
$$
and we are given that $\max_{j\in[2]} \deg_{z_i}\F < D=21$, and 
that 
$\F$ has at most $T=4$ nonzero terms.  Suppose we choose primes 
$p_1 = 5$ and $p_2 = 7$, such that
$$\arraycolsep=0pt
\begin{array}{r clclclcl}
F_1 = F \bmod (z_1^5-1, z_2^5-1) &=& &&
\overbrace{2z_1 
z_2}^{\mathllap{\Phi_1(\ftermarg{1}+\ftermarg{2})}} &+&
\overbrace{2z_1^4}^{{\Phi_1(\ftermarg{3})}} &+& 
\overbrace{4z_1^3,}^{{\Phi_1(\ftermarg{4})}}
\\
F_2 = F \bmod (z_1^7-1, z_2^7-1) &=&
\underbrace{z_1 z_2}_{{\Phi_2(\ftermarg{1})}} &+& 
\underbrace{z_1^6z_2^6}_{{\Phi_2(\ftermarg{2})}} &+& 
\underbrace{2z_1^4 z_2^3}_{{\Phi_2(\ftermarg{3})}} &+&
\underbrace{4z_1^3 z_2^{6}.}_{{\Phi_2(\ftermarg{4})}}
\end{array}
$$
Under the homomorphism $\Phi_1$, the terms $\ftermarg{1}$ and 
$\ftermarg{2}$ {\em collided}.  We call this type of collision an 
{\em exponent collision}.  Keep in mind that we do not know $\F, 
\F_1$, or $\F_2$.  We will construct images of $\F_1$ and $\F_2$ 
in order to recover terms of $\F$.

We choose $\ve_{1 1} = (4,1), \ve_{1 2} = (2,0) \in \ZZ_5^2$, 
such that
$$\arraycolsep=0pt
\begin{array}{r clclclcl}
F_{1 1} = F(x^4, x) \bmod (x^5-1) &=&
\overbrace{2}^{\mathllap{\Psi_{1 1}(\ftermarg{1}+\ftermarg{2})}}
&& &+&
\overbrace{2x}^{\Psi_{1 1}(\ftermarg{3})} &+&
\overbrace{4x^2,}^{{\Psi_{1 1}(\ftermarg{4})}}
\\
F_{1 2} = F(x^2, 1) \bmod (x^5-1) &=&
\underbrace{2x^2}_{\mathllap{\Psi_{1 
2}(\ftermarg{1}+\ftermarg{2})}} 
&& &+&
\underbrace{2x^3}_{{\Psi_{1 2}(\ftermarg{3})}} &+&
\underbrace{4x.}_{{\Psi_{1 2}(\ftermarg{4})}}
\end{array}
$$
If we constructed $\F_{1 1}$ and $\F_{1 2}$, we could suppose 
(correctly) that their terms with coefficient $4$ are images of 
the same single term $f$ of $\F$.  We could then construct the 
exponent $\e$ of $\Phi_1(f)$ as the solution to the linear system
$$
\underbrace{\begin{bmatrix}
4 & 1\\
2 & 0
\end{bmatrix}}_{V_1}\e = \begin{bmatrix} 2 \\ 1 \end{bmatrix} 
\bmod 5,
$$
which gives $\e = (3,0)$, from which we recover the term 
$\Phi_1( \ftermarg{4}) = 4z_1^3$ of $\F_1$.
While $\ftermarg{3}$ also does not collide in either image $\F_{1 
1}$ 
or $\F_{1 2}$, we cannot construct $\Phi_1(\ftermarg{3})$ in the 
same fashion because both $\F_{11}$ and $\F_{12}$ have two terms 
with coefficient $2$.  We call a pair of distinct terms 
$(f,f')$ over all pairs of nonzero terms from the images 
$\F_{ij}, (i,j) \in [2,2],$ a {\em deceptive pair} if $f$ and $f'$ 
share the same 
coefficient but are {\em not} images of the same sum of terms of 
$\F$.  For instance the terms $2$ of $\F_{11}$ and $2x$ of 
$\F_{22}$ form a deceptive pair, as the former is an image of 
$\ftermarg{1}+\ftermarg{2}$, whereas the latter is an image of 
$\ftermarg{3}$.  Similarly, the terms with coefficient $2$ in 
$\F_{11}$ form a deceptive pair, as do the terms $2$ of $\F_{11}$ 
and $2x^2$ of $\F_{12}$.

We similarly choose $\ve_{2 1} = (2, 4), \ve_{2 2} = (1,6) \in 
\ZZ_7^2$, 
such 
that
$$\arraycolsep=0pt
\begin{array}{r clclclcl}
F_{2 1} = F(x^2,x^4) \bmod (x^7-1) &=& &&
\overbrace{3x^6}^{\mathllap{\Psi_{2 1}
(\ftermarg{1}+\ftermarg{3})}} &+&
\overbrace{x}^{\Psi_{2 2}(\ftermarg{2})} &+&
\overbrace{4x^2,}^{{\Psi_{2 1}(\ftermarg{4})}}
\\
F_{2 2} = F(x, x^6) \bmod (x^7-1) &=& &&
\underbrace{2}_{\mathllap{\Psi_{2 2}(\ftermarg{1}+\ftermarg{2})}} 
&+&
\underbrace{2x}_{\Psi_{2 2}(\ftermarg{3})} &+&
\underbrace{4x^4.}_{\Psi_{2 2}(\ftermarg{4})}
\end{array}
$$
Now if have images $\F_{2 1}, \F_{2 2}$ we can reconstruct the 
exponent $\e$ of $\Phi_2(\ftermarg{4})$ by way of the linear system
$$
\underbrace{\begin{bmatrix}
2 & 4 \\ 1 & 6
\end{bmatrix}}_{V_2}\e =
\begin{bmatrix}
2 \\ 4
\end{bmatrix} \bmod 7,
$$
to get a solution $\e = (3,6)$, which gives us the term
$\Phi_2(\ftermarg{4}) = 4z_1^3 z_2^6$ of $\F_2$.  We cannot 
construct 
$\Phi_2(f)$ for any other term $f$ of $\F$ in this fashion, 
because $\ftermarg{1}$ and $\ftermarg{3}$ collide in image $\F_{2 
1}$, and $\ftermarg{1}$ and $\ftermarg{2}$ collide in image $\F_{2 
2}$.  We call these types of collision that depends on our choices 
of $\ve_{2 1}$ and $\ve_{2 2}$ {\em substitution collisions}.  If 
$V_i$ is invertible then any distinct terms of $\F_i$ will not 
collide for at least one vector $\ve_{ij}$, $j \in [n]$.

If we (correctly) suppose our recovered terms $4z_1^3 \in 
\FF_{13}[\Z]/\I_1$ and $4z_1^3 z_2^5 \in \FF_{13}[\Z]/\I_2$ are 
images of the same single term of $\F$, then we can reconstruct 
the exponent $\e$ of the term $\ftermarg{4}$ of $\F$ by way of the 
system of congruences
\begin{align*}
\e \bmod 5 &= (3,0),\\
\e \bmod 7 &= (3,6).
\end{align*}
We use Chinese Remaindering to the solve the system to get $\e = 
(3,20)$, from which we recover term $4z_1^3 z_2^{20}$ of $F$.  Our 
algorithm will choose sufficiently many primes such that, with 
high probability, 
we will be able to construct the exponents of {\em every} term of 
$\F$.

We were unable to recover some of the terms of $F$ because of 
deceptive pairs.  We would like that any terms in a deceptive pair 
are somehow distinguishable as images of distinct sums of terms of 
$\F$.  To that end we can choose $\avec_1 = (6,8) \in \FF_{13}$ 
and construct the corresponding images of $\F(6z_1,8z_2)$.  Note 
that
\begin{align*}
\F(6z_1, 8z_2) &= 9z_1 z_2 + z_1^6 z_2^6 + 8z_1^4 z_2^{10} + 
6z_1^3 z_2^{20},\\
\F(6z_1, 8z_2) \bmod (z_1^5-1, z_2^5-1) &= 10z_1 z_2 + 8z_1^4 + 
6z_1^3,\\
\F(6z_1, 8z_2) \bmod (z_1^7-1, z_2^7-1) &= 9z_1 z_2 + z_1^6 
z_2^6 + 8 z_1^4 z_2^3 + 6z_1^3 z_2^6.
\end{align*}
Again we do not the images listed above, but rather we construct
\begin{align*}
\F_{1 1 1} = \F(6x^4, 8x) \bmod (x^5-1)
&= 10 + 8x + 6x^2,\\
\F_{1 2 1} = \F(6x^2, 8) \bmod (x^5-1)
&= 10x^2 + 8x^3 + 6x,\\
\F_{2 1 1} = \F(6x^2, 8x^4) \bmod (x^7-1)
&= 4x^6 + x + 6x^2,\\
\F_{2 2 1} = \F(6x, 8x^6) \bmod (x^7-1)
&= 10 + 8x + 6x^4.
\end{align*}
Note that the terms of $\F_{ij1}$ have the same exponents as those of
$\F_{ij}$ for $(i,j) \in [2,2]$.  Only the coefficients are 
affected by this additional homomorphism.  However, now we can 
observe, for instance, that the degree-$0$ term of $\F_{1 1 1}$ 
differs from the degree-$3$ term of $\F_{1 2 1}$, such that we now 
can tell the degree-$0$ term of $\F_{1 1}$ and the degree-$3$ 
term of $\F_{1 2}$ form a deceptive pair.  
We say $\avec_1$ {\em reveals} the deceptive pair.  We leave it to 
the reader to check that $\avec_1$ reveals every deceptive 
pair of terms from $\F_{ij}, (i,j) \in [2,2]$.  
We may 
have to 
choose multiple vectors $\avec_k$ in order to reveal all deceptive 
pairs.  We may also have to choose $\avec_k$ over a sufficiently 
large field extension $\FF_{q^u}$.

If all deceptive pairs are 
revealed, then we can collect all terms $f$ of images $\F_{ij}$ 
according to their coefficients and the coefficients of 
their corresponding terms in the images $\F_{ijk}$.  Any such 
collection of terms are then images of the same sum of terms of 
$\F$.

For instance, we can collect the terms with coefficient $2$ in 
$\F_{1 1}, \F_{1 2}$ and $\F_{2 2}$ whose corresponding terms in 
$\F_{111}, \F_{121}$, and $\F_{221}$, respectively, have 
coefficient $8$.  These are exactly the terms of $\F_{ij}$ that 
are images of $\ftermarg{3}$.  Those terms from $\F_{11}$ and 
$\F_{12}$ allow us to construct an exponent $\e=V_1^{-1}(1,2) 
\bmod 5=(4,0)$ of a term of $\F_1$. 
This gives $\e = (4,0)$.  
The corresponding coefficient $2$.  This gives the term 
$\Phi_1(\ftermarg{3})=2z_1^4$ of 
$\F_1$.

Unfortunately, in this same fashion we can now collect terms
$2, 2x^2$, and $2$ of $\F_{11},\F_{12}$ and $\F_{22}$ 
respectively, all images of $\ftermarg{1}+\ftermarg{2}$.  
Considering the pair of such terms in $\F_{11}$ and $\F_{12}$, we 
can construct an 
exponent $\e=V_1^{-1}(0,2) \bmod 5 = 
(1,1)$, which gives a term $2x_1 x_2 \in \FF_{13}[z_1,z_2]/\la 
z_1^5-1,z_2^5-1\ra$.  If $\ftermarg{1}$ and 
$\ftermarg{2}$ had exponent collisions for too many primes $p_i$, 
and we did not detect that the resulting terms of $\F_i$ produced 
were images of a sum of terms of $\F$, then naively we might use 
Chinese Remaindering to then construct an exponent that is {\em 
not} a term of $\F$.

As the partial degrees of $\F$ are at most 20, any exponents $\e 
\neq \e'$ of $\F$ cannot be identical modulo $5$ and $7$.  Thus, 
were {\em every} deceptive pair revealed and we collected terms in 
the manner prescribed, any collection that results in a recovered 
term in both $\F_1$ and $\F_2$ could not be an image of a sum of 
multiple terms of $\F$.  We will use this principle in our 
algorithm in order to distinguish between terms of $F_{ij}$ that 
are images of single terms of $\F$, and those that are images of a 
sum of terms of $\F$.

\subsection{Outline of the paper}

The remainder contains a detailed description and analysis of our algorithm.
Sections \ref{sec:primes}--\ref{sec:div}
give the details and proofs relating to the three randomizations described above:
the degree-reducing primes $p_i$, the univariate substitution 
vectors $\ve_{ij}$,
and the diversification vectors $\avec_k$. Section~\ref{sec:construct} then gives
a full and complete description of the algorithm; a reader uninterested in the
probability analysis may safely skip to this part.
Finally, the proofs of (probabilistic) correctness and running time for our
algorithm are presented in Sections \ref{sec:prob} and \ref{sec:cost}.

\section{Truncating the Degree of Every Variable of $\F$\label{sec:primes}}

In order to interpolate a sparse univariate polynomial $\F$ given 
by an SLP, the usual method is to compute images of the form $\F' = \F 
\bmod (x^p-1)$, where $p$, typically prime, is considerably 
smaller than $D$ (see \cite{ArnGieRoc13,GarSch09,GieRoc11}).  This 
allows us to truncate potential 
intermediate expression swell over the execution of the 
straight-line program.  We typically choose $p$ such that, with 
high probability, the number of terms in collisions in each image 
$\F'$ is either zero or bounded by some fixed proportion $\rho$ of 
$T$.  From this requirement, it follows that most of the terms of 
$\F'$ are images of single terms 
of $\F$.  This is desirable because then most of the terms of 
$\F'$ 
contain ``good'' information about the sparse representation of 
$\F$.  We say two terms of $\F$ {\em collide} in the image $\F'$ 
if their respective images appearing in $\F$ are terms of the same 
degree.  Given a means of sifting good information from ``bad'' 
information (collisions), we can rebuild the sparse representation 
of $\F$ from a sufficiently large set of images.

In the multivariate case, we will more generally consider truncated
images
\begin{equation}\label{eqn:mvar_probe}
\F' = \F(\Z) \bmod (\Z^p-1) \eqdef \f(\z_1, \dots, \z_n) \bmod 
(\z_1^p-1, \dots, \z_n^p-1).
\end{equation}
We let an {\em exponent collision} denote a set of two or more 
terms of $\F$ whose exponents agree under the mapping from $\F$ to 
$\F'$ given by \eqref{eqn:mvar_probe}.  This occurs when there 
exist exponent vectors $\e \neq \e'$ such that
\begin{equation}
(\e - \e') \bmod p = \zero.
\end{equation}

\begin{Lem}\label{Lem:ExpCols}
Let $\F \in \FF_q[\Z]$ be an $n$-variate polynomial 
with $t \leq T$ terms and partial degrees $\deg_{\z_j}(\F) < 
D$, 
for $j \in [n]$.  Let $\mu \in (0,1)$, and let
\begin{equation}
\lambda \geq \max\bigl(21, \tfrac{5}{3}(T-1)\ln D/\mu \bigr).
\end{equation}
Choose a prime $p$ uniformly at random from $(\lambda,2\lambda]$.  
The probability that a fixed term of $\F$ is {\em not} in an 
exponent collision in $\F(\Z) \bmod (\Z^p-1)$ is at 
least $1-\mu$.
\end{Lem}
The lemma is a generalization of Lemma 2.1 of \cite{GieRoc11} and the proof follows similarly.
\begin{proof}[Proof of Lemma \ref{Lem:ExpCols}]
Without loss of generality, we consider the probability that the 
term $\ftermarg{1} = c_1\Z^{\e_1}$ is not in any collision.  Let 
$\mathcal{B}$ comprise the set of all ``bad'' primes $p\in 
[\lambda,2\lambda]$ such that $\ftermarg{1}$ collides with another 
term of $\F$ in the image $\f(\Z) \bmod (\Z^p-1)$.  If 
$\ftermarg{1}$ and $\fterm$ collide modulo $\Z^p-1$, then 
$p$ divides $e_{1j}-e_{\ell j}$ for every $j \in [n]$.  It follows 
that, for each $p\in \mathcal{B}$, $p^n$ divides
\begin{equation}
\prod_{i=2}^t \prod_{j=1}^n (e_{1j}-e_{ij}).
\end{equation}
Thus
\begin{equation}\label{eqn:lambda_def}
\lambda^{|\mathcal{B}|n} \leq \prod_{p \in \mathcal{B}}p^n \leq 
\prod_{i=2}^t \prod_{j=1}^n (e_{1j}-e_{ij}) \leq D^{(T-1)n},
\end{equation}
which gives us $|\mathcal{B}| \leq (T-1)\ln D/\ln \lambda$.  By Corollary 3 of 
\cite{RosSch62}, the total number of primes in the range $(\lambda, 2\lambda]$ 
is at least $3\lambda/(5\ln \lambda)$ for $\lambda \geq 21$.  As
\begin{equation*}
|\mathcal{B}| \leq \tfrac{(T-1)\ln D}{\ln \lambda} \leq \mu\tfrac{3\lambda}{5\ln \lambda},
\end{equation*}
this completes the proof.
\end{proof}

In the subsequent sections, we will show how we can compute a term 
of $\F(\Z)$ modulo $\langle z_1^p-1,\ldots,z_n^p-1\rangle$ 
with probability at least 
$\termProb$.  Therefore our algorithm will
select $m=\lceil 2\log D\rceil$ primes 
$p_i$ with corresponding images $\F_i$, so that any fixed
term $\f$ of $\F$ that is found in at least half of the images
$\F_i$ can be recovered in its entirety.

\section{Substitution Vectors and Substitution 
Collisions}\label{sec:subs}

In this section we will construct a set of images $\F_{ij}$ which 
will allow us to reconstruct {\em some} terms of $\F_i$.  One 
means of interpolating $\F_i$ is via {\em Kronecker substitution}, 
whereby we use a univariate interpolation algorithm to obtain an image
$$
\F_i' = \F_i(x, x^D, \dots, x^{D^{n-1}}) \in \FF_q[x].
$$
An advantage of this map is that it is collision-free, meaning that 
we can obtain every term of $\F_i$ from its image in $\F_i'$.  A 
term of $\F_i$ with exponent $\e$ will result in a term with 
exponent $e = \ssum_{j=1}^n e_j D^{j-1}$ in $\F_i'$, and hence 
$\e$ 
is given by the base-$D$ expansion of $e$.  

But the Kronecker map causes
considerable degree swell, as $\deg(\F_i')$ can be on the order of 
$\deg(\F_i)^n$, exponentially larger than our target degree.
We instead will construct 
$n$ univariate images of $\F_i$,
\begin{equation}
\F_{ij} = \F_i(x^{\ve_{ij}}) \bmod (x^{p_i}-1) \in \FF_q[x]/\la 
x^{p_i}-1 \ra,
\end{equation}
where the $\ve_{ij}$ are randomly chosen vectors from $\FF_q[\z]$, for each
$(i,j) \in [m,n]$.

We say that two terms $f=c\Z^\e$ and $f'=c'\Z^{\e'}$ of $\F_i$ are 
in a {\em substitution collision} if $(\e - \e') \bmod p_i \neq 
\mathbf{0}$, but $(\e - \e')\cdot \ve_{ij} \bmod {p_i} = \mathbf{0}$.  
In which case both terms have an image of degree $\e\cdot 
\ve_{ij} \bmod p$ in $\F_{ij}$.  If $f$ does not collide with any 
other terms of $\F$ in the images $\F_{ij}$, for any $j \in [n]$, 
then we can construct $\e$ as the solution to the linear system
\begin{equation}
\begin{bmatrix} \ve_{i1} \\ \vdots \\ \ve_{in} \end{bmatrix}\e = 
\dvec \bmod p_i,
\end{equation}
provided $V_i = \left[ v_{ijk} \right]_{j,k \in [n]}$ is 
invertible.

\begin{Lem}
Let $p\ge 23$ be prime and
consider a pair of exponents $\e \neq \e' \in \ZZ_p^n$.  Let 
$\ve_1, \dots \ve_n \neq \zero$ be chosen at random from 
$\ZZ_p^n$, Then, with probability exceeding $1-\tfrac{n}{p}$,
the inequality
$\e\cdot \ve_i \neq \e'\cdot \ve_i \bmod p$ holds for each
$i \in [n]$.
\end{Lem}

\begin{proof}
As $\e \neq \e'$, there exists some $j \in [n]$ for which $e_j 
\neq e'_j$.  Without loss of generality assume $e_n \neq 
e'_n$.  Consider a single substitution vector $\ve_i \in 
\FF_{p}^n$.  Note, for any choice of $v_{i1}, \dots, v_{i,n-1}$, 
solving for $v_{in}$ in
\begin{equation*}
(\e - \e')\cdot \ve_i = 0 \bmod p
\end{equation*}
gives
\begin{equation*}
  v_{in} = \frac{\sum_{j=1}^{n-1} (e_j-e_j')v_{ij}}{e'_{n}-e_{n}} 
\bmod p.
\end{equation*}
That is, $v_{in}$ is uniquely determined given any choices for
the other elements in $\ve_i$.
Thus precisely a proportion $1/p$ of choices of substitution 
vectors $\ve \in \ZZ_p^n$ will cause a substitution collision 
between the terms with exponents $\e$ and $\e'$. By the 
union bound, the probability that any one of a random choice of $n$ 
vectors $\ve_1, \dots, \ve_n$ results in a substitution collision 
between $\e$ and $\e'$ is at most $n/p$.
\end{proof}

We also require that the matrix $V$ is invertible.  In a 
practical setting, if $V$ is not invertible, 
we might reasonably just choose $n^2$ new random entries for $V$
and try again.
But for the purposes obtaining fast deterministic running 
time in a Monte Carlo setting, if $V$ 
is singular, our algorithm will merely ignore the images 
$\F_{ij}$.  By 
(\cite{Dix01}, Part II, Chapter 1, Theorem 99), we have that the 
probability that a matrix chosen at random from $\ZZ_p^{n 
\times n}$ is invertible is $\prod_{i=1}^n (1-1/p^i)$.  For $p 
\geq 23$,
\begin{equation*}
\sprod_{i=1}^n(1-1/p^i) \geq \sprod_{i\geq 1}(1-1/23^i) \approx 
0.95463 \geq 19/20.
\end{equation*}
By the union bound we get the following lemma.

\begin{Lem}\label{Lem:SubsCols}
Consider a pair of exponents $\e \neq \e' \in \ZZ_p^n$.  Let 
$\ve_1, \dots \ve_n \in \ZZ_p^n$ be chosen uniformly at random, 
where $p \geq 23$. Let $V$ be the matrix whose rows are 
given by $\ve_1, \dots, \ve_n$, chosen at random from $\ZZ_p^n$.  
Then with probability at least $1-n/p-1/20$, $V$ is invertible and 
$\e\cdot \ve_j \neq \e'\cdot \ve_j 
\bmod p$ for every $j \in [n]$.
\end{Lem}

We can use Lemmata \ref{Lem:ExpCols} and \ref{Lem:SubsCols} to 
bound the probability of a term of $\F$ being involved in either 
an exponent collision or a substitution collision by $1-\mu/2$ and 
$\tfrac{19}{20}-\mu/2$, respectively, with an appropriate choice of $\lambda$.  
By the union bound this gives the 
following 
Corollary.
\begin{Cor}\label{Cor:term}
Let $p_i$ be chosen at random from $(\lambda, 2\lambda]$, where
\begin{equation}
\lambda = \max\left(21, \tfrac{5}{6}(T-1)\ln D/\mu, 2n/\mu\right),
\end{equation}
Let $\ve_1, \dots, \ve_n \in \ZZ_{p_i}^n$, be chosen uniformly at 
random, and let $V$ be the matrix whose rows are given by the 
$\ve_j$.  Fix a term $\f$ of $\F$.  Then, with probability 
exceeding $\tfrac{19}{20}-\mu$, $V$ is invertible and $\f$ does not 
collide with any other term of $\F$ in the images $\F_{ij}$.
\end{Cor}
In other words, with probability greater than $19/20-\mu$ we can 
obtain $c$ and $\e \bmod p_i$ for a term $f = c\Z^\e$ of $\F$, 
from $n$
images $F_{ij}$. This is provided that we can identify the terms from
the images $\F_{ij}$ that correspond to $f$. We accomplish this task
in the next section, by showing how to
group terms that are images of the same term, or sum of terms, of 
$\F$.

\section{Diversification\label{sec:div}}

We need a means of collecting terms amongst the images 
$\F_{ij}=\F(x^{\ve_{ij}}) \bmod (x^{p_i}-1)$ which are images of the 
same term, or sum of terms, of $\F$.  Consider a pair of images
\begin{align*}
\F^{(1)}  &=  \F(x^{\ve'_1}) \bmod (x^{p'_1}-1),&
\F^{(2)} &=  \F(x^{\ve'_2}) \bmod (x^{p'_2}-1),
\end{align*}
where $p'_k$ is prime and $\ve'_k \in \ZZ_{p'_k}$ for $k = 1,2$.  
Suppose $\F^{(k)}$ has a term $f_k = bx^{d_k}$, for $k=1,2$.  As 
these terms share the same coefficient $b$, it is possible that they are 
images of the same term of $\F$. However, they could also be images of
two different terms of $\F$ that happen to have the same 
coefficient.
Moreover, it is possible that one or both are images of a \emph{sum} of
terms from $\F$.
In particular, $f_k$ is image of the sum of 
all terms $c\Z^\e$ of $\F$ such that $\e \cdot \ve_k' \bmod p_k' = 
d_k$, for $k=1,2$.
Let $h_1, h_2$ be the 
respective sums of such terms, 
i.e.,
\begin{align*}
h_k = \hspace{-23pt}\sum_{\substack{\ell \in [t]\\ \e_\ell\cdot 
\ve'_k \bmod 
p'_k = d_k}}\hspace{-23pt}c_\ell\Z^{\e_\ell},
\quad k=1,2,
\end{align*}
and let $h = h_1-h_2$.  The coefficient of the $x^{d_k}$ term of 
$\F^{(k)}$ is 
$h_k(1,1,\dots, 1)$, $k=1,2$.  These terms share the same 
coefficient if and only if $h(1,1,\dots,1)=0$.  If 
$h$ is not the zero polynomial but 
$h(1,\dots,1)=0$, then 
we might erroneously believe that $f_1$ and $f_2$ are images of 
the same term or sum of terms of $\F$.  We call such an 
unordered pair of terms $\{f_1, f_2\}$ a {\em deceptive 
pair}.

Now consider $\avec \in \left(\FF_{q^u}^*\right)^n$ chosen 
uniformly at random, 
where $u \geq 1$, 
and the images
\begin{align*}
\widetilde{F}^{(1)} &= \F(\avec x^{\ve_1'}) \bmod (x^{p_1'}-1), &
\widetilde{F}^{(2)} &= \F(\avec x^{\ve_2'}) \bmod (x^{p_1'}-1),
\end{align*}
Now the degree-$d_k$ term of $\widetilde{F}^{(k)}$ is 
$h_k(\avec)$.  
If we observe $h(\avec)\neq 0$, then we can conclude  
that $f_1$ and $f_2$ were not images of the the same sum of 
terms of $\F$.  In this instance we say $\avec$ {\em reveals} the 
deceptive pair $\{f_1,f_2\}$.  We will choose 
nonzero entries for $\avec \neq \zero$ from a field extension 
$\FF_{q^u}$, where $u=\udef$.  As 
the total 
degree of $h$ is less than $nD$, then by the Schwartz-Zippel 
Lemma \cite{Zip80}, the probability that 
$h(\avec) = 0$ is at most $nD/q^u$.  By our choice of $u$ 
the probability that $h(\avec)=0$ is at most $\tfrac{1}{2}$.

If we choose $s$ random vectors $\avec_1, \dots, \avec_s 
\in (\FF_{q^u}^*)^n$ independently and uniformly at random, then 
probability that none of the $\avec_i$ reveal a given deceptive 
pair is less than $2^{-s}$.  Choosing
$$
s=\lceil \log\tfrac{1}{\mu} + 2\log\ell + 2\log n + 2\log T\rceil,
$$
gives a probability bound of
\begin{equation*}
\left(\tfrac{1}{2}\right)^{\lceil \log\tfrac{1}{\mu} + 2\log m + 
2\log n + 2\log T\rceil } = \tfrac{1}{\mu}(m^2n^2T^2)^{-1}.
\end{equation*}
The images $\F_{ij}$, $(i,j) \in [m,n]$, have collectively at most 
$mn T$ terms, thus fewer than $m^2n^2T^2$ deceptive pairs 
may occur.  By the union bound we get the following lemma, which 
shows that the vectors $\avec_k, 
k\in[s]$ reveal {\em all} possible deceptive pairs with high 
probability.

\begin{Lem}\label{Lem:div}
Let $\mu \in (0,1)$, $u=\udef$, and 
$s=\lceil \log\tfrac{1}{\mu} + 2\log m + 2\log n + 2\log T 
\rceil$.  Choose $\avec_1, \dots, \avec_s$ independently and 
uniformly at random from $\left( \FF_{q^u}^* \right)^n$.  Then 
$\avec_1, \dots, \avec_m$ reveal every deceptive pair amongst all 
pairs of terms from images $\F_{ij}$, $(i,j) \in [m,n]$, with 
probability greater than $1-\mu$.
\end{Lem}

With this lemma we now have the tools required in order 
to reconstruct the terms of $\F$.

\begin{procedure}[htbp]\footnotesize
\caption{SparseInterpolate($\SLP, D, T$)}\label{proc:si}
\KwIn{$\SLP$, an SLP computing $\F = \ssum_{\ell=1}^t 
c_\ell\Z^{\e_\ell} 
\in \FF_q[\z_1, \dots, \z_n]$; $D > \max_{j \in [n]}\deg_{\z_i}\F$; $T 
\geq t$}
\KwOut{ $\F^*$, a sparse representation of $\F$, with probability $\geq 
3/4$}

\noindent\textbf{precomputation}\\
$\mathcal{D} \leftarrow \createdict()$ \;
$m \leftarrow \max(6, \lceil 2\log D\rceil, \lceil (25/8)\ln(4T)\rceil)$\;
$\lambda \leftarrow \max(21, \tfrac{95}{3}(T-1)\ln D, 80n, \tfrac{10}{3}m 
\ln m)$\;
$s \leftarrow \lceil \log 40 + 2\log m + 2\log n + 2\log T \rceil$\;
$u \leftarrow \udef$\;
Choose the following independently and uniformly at 
random:\\
\indent $\bullet$ $p_i$, a prime in $(\lambda, 2\lambda]$ for $i \in [m]$\;
\indent $\bullet$ $\ve_{ij} \in \ZZ_{p_i}^n$, for $i, j \in [m,n]$\;
\indent $\bullet$ $\avec_k \in \left(\FF_{q^u}^*\right)^n$, for $k \in 
[s]$\;
\noindent\textbf{begin}\\

\For {$i \in [m]$\label{line:start_cong}}
{
$V_i \leftarrow [ v_{ijk} ]_{j,k=1}^n \in \FF_{q}^{n \times n}$\;
\lIf{$V_i$ is invertible}{ Compute $V_i^{-1}$ \textbf{else} \continue 
}\label{line:inv}
$\localdict \leftarrow \createdict()$\;
  \For{$j \in [n]$\label{line:jloop}}{
    $\F_{ij} \leftarrow \F(\z^{\e \cdot \ve_{ij}}) \bmod 
    (\z^{p_i}-1)$\;
    \lFor { $k \in [s]$ }{
      $\F_{ijk} \leftarrow \F(\avec_k\z^{\e \cdot \ve_{ij}})$
    }
    \ForEach{ \label{line:innerloop} nonzero term $b_0\Z^d$ of 
    $\F_{ij}$}{
      \lFor{ $k \in [s]$ }{
        $b_k \leftarrow $ coefficient of the degree-$d$ term of 
        $\F_{ijk}$
      }
      $\localdict.\appendto(\bee,d)$ \; \label{line:jloopend}
    }
  }
  \For{$(\bee, \dvec) \in \localdict.\getitems()$}{
    \lIf{$|\dvec| \ne n$}{ \continue}
    $\e \leftarrow V_i^{-1}\dvec \bmod p_i$\; 
    \label{line:make_cong}
    $\dict.\appendto(\bee,(\e \bmod p_i))$ \; \label{line:add_to_cong}
  }\label{line:end_cong}
}

$\F^* \leftarrow 0 \in \FF_q[\Z]$\;
\label{line:start_terms}
\For{ $(\bee,\mathcal{C} \in \dict.\getitems()$ }{
   \lIf{ $|\mathcal{C}| < m/2$ }{ \continue }
   $\e \leftarrow$ solution to set of congruences $\mathcal{C}$\;
   \label{line:crt}
   $\F^* \leftarrow \F^* + b_0\z^\e$\; \label{line:end_terms}
}
\Return $F^*$\;
\end{procedure}

\section{Description of the Algorithm\label{sec:construct}}

Our algorithm is given by Procedure~\ref{proc:si} on page \pageref{proc:si}.  
It is divided into two parts: the precomputation phase which
defines various parameters and chooses all necessary random values,
and the actual algorithm which performs the evaluations
of the images to eventually reconstruct $F$. We note that the
precomputation phase does \emph{not} dominate the cost of the
algorithm, and hence could be considered simply ``computation''
with no asymptotic effects.

After the precomputation, we use the straight-line
program to construct images
\begin{align*}
\F_{ij} &= \F(x^{\ve_j}) \bmod (x^{p_i}-1),\quad\text{and}\\
\F_{ijk} &= \F(\avec_k x^{\ve_j}) \bmod (x^{p_i}-1),\quad (i,j,k) 
\in [m,n,s],
\end{align*}
for every $i$ such that $V_i$ is invertible.  If $V_i$ is not 
invertible, it will be impossible to uniquely 
recover terms of $\F_i$, so we 
continue (line \ref{line:inv}).  
Our analysis will show that the diversification vectors
$\avec_k, k\in[s]$ are sufficient to reveal all ``deceptive pairs''
of colliding terms in these images, in the language of the
previous sections.

We use these images to construct congruences of the form
$(\e \bmod p_i)$, $\e \in \ZZ_{p_i}^n$ (lines 
\ref{line:start_cong}--\ref{line:end_cong}).  Each congruence $(\e 
\bmod p_i)$ is 
constructed 
as a solution to a linear system $V_i\e = \dvec \bmod p_i$.  These 
congruences
are each uniquely associated with a vector of coefficients 
$\bee = (b_0, \dots, b_{s})\in \FF_{q^u}^{s+1}$, where
$b_0\in\FF_q$ is the coefficient in the base field that appeared
in $\F_{ij}$.  Then, for any $\bee$ that has a
sufficiently large set of congruences, we can recover the actual
exponent vector $\e\in\ZZ^n$ by way of 
Chinese Remaindering and add $b_0\Z^\e$ to $\F^*$, a sparse 
representation for $F$ (lines 
\ref{line:start_terms}--\ref{line:end_terms}).

In the $i$-th iteration of the for loop starting on line 
\ref{line:start_cong}, 
we build the set of tuples $\tupleset_i$ 
comprised of all $(\bee, \dvec)$, $\bee = (b_0, \dots, b_s) \in 
\FF_{q^u}^{s+1}$, such that $\F_{ij}$ has a 
term $b_0x^{d_j}$ and $\F_{ijk}$ has a term $b_kx^{d_j}$ for all 
$(j,k) \in [n,s]$.  For each such tuple, we 
construct a congruence $(\e \bmod \p_i)$, where $\e \in 
\ZZ_{p_i}^n$ is the solution to the linear system
$$
V_i \e = \dvec \bmod p_i.
$$

We build this set of tuples using a dictionary $\localdict$, whose 
keys are $\bee \in \FF_{q^u}^{s+1}$ and whose values are degree 
vectors $\dvec \in \ZZ_{p_i}^n$.  We build the degree vectors 
$\dvec$ iteratively.  During $j$-th iteration of the for loop on 
line \ref{line:jloop}, 
we construct all the tuples $\bee$ such that $\F_{ij}$ contains a 
nonzero term $b_0x^d$ and $\F_{ijk}$ has a term $b_kx^d$ for $k 
\in [s]$. 

Each iteration adds one more entry to the degree vector
$\dvec$, so that for any term which appeared uncollided in every image
$\F_{ij}, j\in[n]$, the corresponding degree vector $\dvec$ will
have length exactly $n$ by the time the loop is finished after
line~\ref{line:jloopend}. Denote by $\tupleset_i$ the set of
all such $(\bee,\dvec)$ tuples for which $|\dvec|=n$.

For each $(\bee, \dvec) \in 
\tupleset_i$, we recover $\e = V_i^{-1}\dvec \bmod p_i, \e \in 
\ZZ_{p_i}^n$ (line \ref{line:make_cong}).  For any key $\bee$ we 
may recover up to $m$ such 
congruences, one for each $i \in [m]$.  For each $\bee$ we store 
the set $\mathcal{C}$ of such congruences in a second dictionary 
$\dict$ (line \ref{line:add_to_cong}).

The two dictionaries $\localdict$ and $\dict$ are maps from
tuples $\bee\in\FF_{q^u}^{s+1}$ to lists of values. Under the
reasonable assumption that there is a consistent and computable
ordering on the base field $\FF_q$, these dictionaries can be
implemented as any balanced search tree, such as an AVL tree or
red-black tree.

In particular, the dictionary data structures should support the
following operations. The running times are expressed in bit cost,
which is affected by the fact that keys in $\FF_{q^u}^{s+1}$ have
$O(su\log q)$ bits each.
\begin{itemize}
\item $\createdict()$ constructs a new, empty dictionary.
  The running time is $O(1)$.
\item $\dict.\appendto(\key, \val)$ first searches to see if
  $\key$ is already in the tree. If it is, $\val$ is appended
  to the end of the list associated with $\key$. Otherwise,
  $\key$ is added to the tree and associated with a new list
  containing $\val$.
  The bit cost of this operation, for keys in $\FF_{q^u}^{s+1}$,
  is $O(\log |\dict| \cdot s u \log q)$.
\item $\dict.\getitems()$ iterates over all $(\key, \lst)$ pairs
  of keys and lists of values that are stored in the tree.
  The bit cost of this operation is linear in the total size of
  the dictionary and its keys,
  $O(|\dict| \cdot s u \log q)$.
\end{itemize}

Provided every deceptive pair is revealed, every key $\bee$ of 
$\dict$ uniquely corresponds to a fixed, nonempty sum of terms of 
$\F$.  We are interested in those corresponding to exactly one 
term of $\F$.  With high probability, for any term $f$ of $\F$, we 
can construct 
the image of $f$ in $\F_i$ for at least half of the $i \in [m]$.  
In other words, with high probability, any $\bee$ 
corresponding to a term $f$ of $\F$ should be associated in 
$\dict$ with a set of at least $m/2$ congruences.  By setting $m 
\geq 2\lceil \log D\rceil$, any $\bee$ corresponding 
to a collision of terms of $\F$ will produce a set of less than 
$m/2$ congruences.  Otherwise, this collision would contain 
terms $c\Z^\e \neq c'\Z^{\e'}$ of $\F$ such that $(\e-\e') \bmod 
p_i=0$, for over $m/2$ primes $p_i$.  This gives a contradiction 
as the product of such primes is at least $D$ but the partial 
degrees of $\F$ are less than $D$.

For any key $\bee$ associated with a set of at least $m/2$ 
congruences $(\e \bmod p_i)$, we reconstruct 
$\e \in \ZZ_D^n$ by way of Chinese Remaindering (line 
\ref{line:crt}).  We then add a term $b_0\Z^\e$ to a sparse 
polynomial $\F^*$ (line \ref{line:end_terms}).  
Provided each probabilistic step of the algorithm succeeded, this 
sum of such terms then comprises the sparse representation of $\F$.

\section{Probability Analysis\label{sec:prob}}

The Procedure~\ref{proc:si} sets the following four parameters in order
to guarantee Monte Carlo-type correctness:
\begin{align}
m &= \mdef,\label{eqn:mdef}\\
\lambda &= \lambdadef,\label{eqn:lambdadef}\\
s &=\sdef,\\
u &=\udef.\label{eqn:udef}
\end{align}

The setting of these parameters is explained by the following lemma.

\begin{Lem}\label{Lem:prob}
Procedure \ref{proc:si} correctly outputs a sparse representation of $\F$ with 
probability at least $\tfrac{3}{4}$.
\end{Lem}

\begin{proof}
We first require that the interval $(\lambda,2\lambda]$ contains 
at least $m$ primes.  By \cite{RosSch62}, the total number of primes in 
$(\lambda,2\lambda]$ is at least $3\lambda/(5\ln\lambda)$.  
Because
$\lambda > \tfrac{10}{3}m \ln m$,
\begin{align*}
3\lambda/(5\ln\lambda) = m\ln(m^2)/\ln\left( (10/3)m\ln m \right) \geq m.
\end{align*}
The last inequality above holds whenever 
$(10/3)m \ln m \leq m^2$, which is true for all $m \geq 6$.  

Next, recall that in our notation, the polynomial we wish to interpolate is
written term-wise as
$\F = \sum_{\ell \in [t]} c_\ell \Z^{\e_\ell}.$
The algorithm works by computing $m$ images $\F_i = \F \bmod
(\Z^{p_i}-1)$. 

Fix a single term $c_\ell\Z^{\e_\ell}$, for an arbitrary $\ell\in[t]$.
If (1) the random coefficient vectors $\avec_k$ reveal all
deceptive terms as in Lemma~\ref{Lem:div}; (2) the matrix of
exponent substitutions $V_i = (v_{ijk})_{j,k\in[n,n]}$ is is invertible;
and (3) the term $c_\ell\Z^{\e_\ell}$ does not collide with any other
terms modulo $p_i$, then the exponent vector $\e_\ell$ will be recovered
modulo $p_i$ at that step.

Lemma~\ref{Lem:div} guarantees condition (1) with probability at least
$1-\mu$, and Corollary~\ref{Cor:term} guarantees (2) and (3), for a
fixed term $c_\ell\Z^{\e_\ell}$ and prime $p_i$, with probability at
least $\tfrac{19}{20}-\mu$. By the union bound, the probability that all
three conditions hold, and therefore that we recover the exponents of
this fixed term modulo any given prime $p_i$, is at least
$\tfrac{19}{20}-2\mu$. By setting $\mu = \tfrac{1}{40}$,
the probability of failure for any fixed term index $\ell$ and prime
index $i$ is at most $\tfrac{1}{10}$.

Observe that the number $m$ of primes $p_i$ is at least $2\log_2 D$, so
that if the exponent vector $\e_\ell$ is recovered modulo any fraction
$m/2$ of the primes, then there is sufficient information to recover the
actual exponents $\e_\ell$ over $\ZZ$ at the end of the algorithm. The
preceding paragraph shows that the \emph{expected} number of primes
$p_i$ for which we can recover the exponent vector modulo $p_i$ is at
least $9m/10$. Hoeffding's inequality provides a way to bound the
probability that the actual number of primes that allow us to recover
$\e_\ell$ is less than $m/2$, much smaller than the expected value of
$9m/10$.

Define the random variable $X_i$ to be 1 if the exponent vector
$\e_\ell$ is recovered modulo $p_i$, and 0 otherwise. Hence
$\mathbb{E}[X_i] = \tfrac{9}{10}$. We want to know the probability that
$\sum_{i\in m} X_i \ge \tfrac{m}{2}$. Since the primes $p_i$ are chosen
uniformly at random, without replacement, Theorems 1 and 4 from
\citep{Hoe63} tell us that
\[\Pr\left[\sum_{i\in m}X_i < \frac{m}{2}\right] <
\exp\left(\frac{-8m}{25}\right).\]

As there are $T$ total terms, the union bound tells us that the
probability of recovering \emph{every} term in at least $m/2$ of the
images $\F_i$ is at least $1 - T\exp(-8m/25)$, which is at least
$\tfrac{3}{4}$ since we choose $m \ge (25/8)\ln (4T)$.
\end{proof}

We employ a meta-algorithm in order to interpolate $\F$ with an arbitrarily 
small probability of failure $\epsilon < 1$.  
One iteration of the algorithm succeeds with
probability at least 
$3/4$.  Thus, if we run the algorithm $r$ times producing outputs $\F^*_1, 
\dots, \F^*_r$, then by Hoeffding's inequality \citep[][Theorem
1]{Hoe63}, the probability that $\F_i 
\neq \F$ for 
at least half 
of the $\F^*_i$, $i \in [r]$, is less than $e^{-r/8}$.  Setting 
$r=8\ln(1/\epsilon)$ makes this probability less than $\epsilon$.  We thus 
merely run the algorithm 
$\lceil 8\ln \tfrac{1}{\epsilon}\rceil$ times and return the polynomial $\F$ that appears 
most frequently.
If no such $\F$ appears more than half the time, the meta-algorithm fails.

\section{Cost Analysis\label{sec:cost}}

We give a ``soft-oh'' cost analysis of the algorithm, whereby we 
ignore possible additional logarithmic factors in the cost.  
As the final cost is polynomial in $\log D, T, n$, and $\log q$, 
we ignore poly-logarithmic factors of these values.  We let 
$\kappa$ denote any term that is poly-logarithmic in $nT\log D\log 
q$, i.e., any term that is bounded by
$\log^{\bigoh(1)}(nT\log D\log q)$, to simplify the cost analysis of intermediate steps.

Equations 
\eqref{eqn:mdef}-\eqref{eqn:udef} give us 
\begin{align*}
m & \in \bigoh(\log D + \log T) \subseteq \softoh(\kappa \log D),\\
\lambda & \in \bigoh\left(\left(T + \log\log D\right)\log D + 
n\right) 
  \subseteq \softoh(T\log D + n)\\
s &\in \bigoh(\log\log D + \log n + \log T) \subseteq \softoh(\kappa),\\
u & \in \bigoh(1 + (\log D + \log n)/\log q) \subseteq \softoh(1 + \kappa \log
D/\log q).
\end{align*}

Observe also that each of $\log m, \log \lambda, \log s, \log u$ is
$\softoh(\kappa)$, and therefore does not affect the overall soft-oh
analysis.

\subsection{Cost of Precomputation}

The precomputation steps involve setting up the fields and constants
that the algorithm uses, as well as making all the necessary random
choices. We will repeatedly need to choose a single item uniformly
at random from a finite set $S$. Observe that such a choice can be
made in $\bigoh(\log |S|)$ time and using $\bigoh(\log |S|)$ random bits by,
for example, defining an enumeration of the set and choosing a random
index between 0 and $|S|-1$.

\subsubsection{Cost of Generating Primes}

The algorithm requires selecting uniformly at random a list of
$m$ primes $p_i$ in the range $(\lambda, 2\lambda]$.  It is 
possible to generate all primes in $(\lambda, 
2\lambda]$ with 
\begin{equation}\label{eqn:chooseprimes}
\softoh(\lambda)=\softoh(T\log D + n)
\end{equation}
bit operations using a sieve method, e.g., the wheel sieve 
\cite{Pri82}. From the previous discussion, choosing $m$ of these
at random from the generated list would cost
$\bigoh(m\log\lambda)$, which is $\softoh(\kappa\log D)$ and thus dominated
by the cost of sieving.

A more practical approach might be to
select each such prime 
probabilistically, with expected bit-cost poly-logarithmic in 
$\lambda$, by selecting integers $p$ at random from $(\lambda, 
2\lambda]$, $p$ perhaps not a multiple of a small prime, and then 
running probabilistic primality test on it, e.g. Miller-Rabin.  
But as that would complicate our probabilistic analysis (namely by
increasing the probability of failure), we will assume for the purposes
of analysis that the sieving method is used to choose primes.

\subsubsection{Cost of constructing a field extension 
$\FF_{q^u}$.}  
In order to select vectors $\avec \in \FF_{q^u}$ at random, we first need to 
construct a representation for $\FF_{q^u}$.  In particular, we 
need to construct an irreducible polynomial of degree $u$ over 
$\FF_q$.  Per \cite{Sho94}, this may be done in $\softoh( u^2 + 
u\log q)$ operations in $\FF_q$, for a total bit complexity of
$\softoh(u\log q (u+\log q))$.

Because $u \in \softoh(1 + \kappa\log D/\log q)$, we can see that
$u\log q \in \softoh(\log q + \kappa\log D)$. Furthermore, we only
work in an extension provided $q \leq 2nD$, which means that
in this case we always have $\log q \in \softoh(\kappa \log D)$.
Hence $u\log q$ and $(u + \log q)$ are both $\softoh(\kappa\log D)$,
and the entire cost of this step is less than 
$\softoh(\kappa\log^2 D)$.

\subsubsection{Cost of selecting $\avec_k$ and $\ve_{ij}$.} 

As $\FF_{q^u}$ is a finite set of size $q^u$, the cost of selecting
a single element from that field is $\bigoh(u\log q)$, which is
$\softoh(\log q + \kappa\log D)$. Our algorithm requires $s$ length-$n$
vectors $\avec_k \in \FF_{q^u}^n$, for a total cost of
\begin{equation}\label{eqn:choose_a}
\bigoh(snu\log q) \subseteq \softoh\left(\kappa n\left(\log D + \log q
\right)\right)
\end{equation}
bit operations.

The algorithm also requires choosing $mn$ size-$n$ vectors 
$\ve_{ij} \in \ZZ_{p_i}^n$, where each $p_i 
\in \bigoh(\lambda)$.
The cost of selecting these vectors is
\begin{equation}\label{eqn:choose_ve}
\bigoh( m n^2 \log \lambda) \subseteq \softoh( \kappa n^2 \log D).
\end{equation}

Summing up all the precomputation costs above and removing
extraneous
factors of $\kappa$ gives a total of
\begin{equation}
\softoh\left(\left(T + n^2 + \log D\right) \log D + n\log q\right)
\end{equation}
bit operations.

Furthermore, considering the costs \ref{eqn:chooseprimes}, 
\ref{eqn:choose_a}, and \ref{eqn:choose_ve} of selecting randomly 
chosen primes and 
vectors, we can also bound the bits of randomness required by 
Procedure~\ref{proc:si} by
\begin{align*}
\bigoh(m\log \lambda + snu\log q + mn^2\log\lambda) =
\softoh\left(n\log T\left(n\log D + n\log T + \log q\right)\right).
\end{align*}
We give this as a lemma.
\begin{Lem}\label{lem:random}
Procedure~\ref{proc:si} requires 
$\softoh\left(n\log T\left(n\log D + n\log T + \log q\right)\right)$
bits of randomness.
\end{Lem}

\subsection{Cost of Producing Images}

The algorithm produces $\bigoh(mns)$ images in rings 
$\FF_{q^u}[x]/(x^p-1)$, $p \in \bigoh(\lambda)$.  Per Claim 
\ref{claim:probe}, this cost is $\softoh(mns\cdot Lpu\log q)$
bit operations, or
\begin{equation}\label{eqn:imgcost}
\softoh(Ln(T\log D+n)(\log D + \log q)\log D).
\end{equation}
This dominates the cost of precomputation.

\subsection{Cost of Constructing and Accessing Dictionaries}

Observe that as every image $\F_{ij}$ has at most $T$ nonzero terms, we 
run the for loop on line \ref{line:innerloop} of Procedure~\ref{proc:si} at 
most $T$ times 
consecutively.  
It follows that $|\localdict| \leq nT$ at any 
point of the algorithm.  Counting loop iterations, this implies 
that Procedure~\ref{proc:si} runs $\localdict.\appendto$ at 
most $mnT$ times, in addition to running $\localdict.\createdict$ 
and $\localdict.\getitems$ $m$ times.

It follows from the discussion in section \ref{sec:construct} that 
the cost of these operations will total
\begin{align*}
\softoh( mnT\log|\localdict|su\log q + m|\localdict|su\log q) = 
\softoh( mnTsu\log q).
\end{align*}

As $|\localdict| \leq nT$, we run $\dict.\appendto$ at most $nT$ 
times for every iteration of the outer for loop beginning on line 
\ref{line:innerloop}.  It follows that $|\dict| \leq mnT$ over the 
execution of the algorithm, and that $\dict.\appendto$ is run at 
most $mnT$ times.  These operations yield a cost of $\softoh( mnT 
\log|\dict| su\log q)=\softoh(mnTsu\log q)$.  The soft-oh cost 
due to running $\dict.\getitems$ once is similar.  Thus the cost 
for all dictionary operations becomes
$$
\softoh( mnTsu \log q) = \softoh(nT(\log D + \log q)\log D).
$$
This cost is also absorbed by the cost \eqref{eqn:imgcost} of constructing images.

\subsection{Cost of Solving Linear Systems}

For each of the $m$ primes $p_i$, we need to invert an $n\times n$ 
matrix $V_i$ with entries in $\ZZ_{p_i}$.  This requires 
$\bigoh(n^\omega)$ operations in $\ZZ_{p_i}$, where $p_i \in \bigoh( 
\lambda)$ (see, e.g., Prop. 16.6 in \citet{clausen1997algebraic}).  
This bit-operation 
cost becomes
\begin{equation}\label{inv:cost}
\softoh( \log(\lambda)m n^{\omega} ) = \softoh( \kappa(\log 
D)n^{\omega}).
\end{equation}

In addition,
for every $i \in [n]$ we have to compute some number of vectors $\e 
\bmod p_i$ as products $V_i^{-1}\dvec \bmod p_i$ for various 
$\dvec \in \ZZ_{p_i}^n$ on line \ref{line:make_cong}. 
As $\localdict.\appendto$ is run at most $nT$ times in 
a single iteration of the outer for loop starting on line 
\ref{line:start_cong}, then $\localdict$ can only have at most $T$ 
entries with values $\dvec$ with length $n$. 

These linear system solutions $V_i^{-1}\dvec \bmod p_i$
can be performed more efficiently via blocking, whereby we 
multiply $V_i^{-1}$ by $n$ vectors $\dvec$ at a time using fast 
matrix multiplication.  This entails 
$\softoh\left( n^\omega \left\lceil\tfrac{T}{n}\right\rceil \right)$
arithmetic operations in $\ZZ_{p_i}, p_i \in 
\bigoh(\lambda)$.
Thus, using 
this blocking strategy the total cost of the linear system 
solving over every iteration of the outer for-loop becomes 
\begin{align*}
\softoh\left( n^\omega m\left\lceil \frac{T}{n} \right\rceil\right) &\in \softoh( 
mn^{\omega}(T/n+1)),\\
&=\softoh( n^{\omega-1}mT + n^\omega m ),\\
&=\softoh( n^{\omega-1}T\log D + n^{\omega}\log D),
\end{align*}
which dominates the cost \eqref{inv:cost} of computing the inverses.

\subsection{Cost of Constructing Terms}

To construct the terms of $\F$, we have to construct $T$ exponents 
$\e \in \ZZ_D^n$, from sets of at most $m$ congruences.  
Constructing one entry $e_j \in [0, D)$ of one exponent $\e$ by 
the Chinese Remainder algorithm entails $\bigoh( \log^2 D)$ bit 
operations (Thm. 5.8, \cite{MCA}).  Doing this for at most $Tn$ 
vector entries yields a total cost of $\bigoh( nT\log^2 D)$.  
Again this cost is absorbed by the cost \eqref{eqn:imgcost}
of constructing images.

From the previous subsections, we see that the total cost
of Procedure~\ref{proc:si} is dominated by the cost of constructing 
the images of $\F$, and the cost of solving linear systems.  We 
state the total bit-cost of the algorithm as a lemma.
\begin{Lem}\label{Lem:cost}
Procedure~\ref{proc:si} entails a bit-operation cost of
\begin{equation*}
\softoh\left( Ln(T\log D + n)(\log D + \log q)\log D + 
n^{\omega-1}T\log D + n^\omega \log D \right).
\end{equation*}
\end{Lem}
Combining Lemmata \ref{Lem:prob} and \ref{Lem:cost} gives Theorem 
\ref{Thm:Main}.

\section{Acknowledgements}
The first author thanks Mustafa Elsheikh for helpful discussion
related to section \ref{sec:subs}.  The first author further
acknowledges the support of the National Sciences and Engineering 
Research Council of Canada (NSERC).
The third author is supported by National Science Foundation (NSF)
award no.\ 1319994,
``AF: Small: RUI: Faster Arithmetic for Sparse Polynomials and
Integers''.

\bibliographystyle{elsarticle-harv} 

\begin{thebibliography}{30}
\expandafter\ifx\csname natexlab\endcsname\relax\def\natexlab#1{#1}\fi
\expandafter\ifx\csname url\endcsname\relax
  \def\url#1{\texttt{#1}}\fi
\expandafter\ifx\csname urlprefix\endcsname\relax\def\urlprefix{URL }\fi

\bibitem[{Arnold et~al.(2013)Arnold, Giesbrecht, and Roche}]{ArnGieRoc13}
Arnold, A., Giesbrecht, M., Roche, D.~S., 2013. Faster sparse interpolation of
  straight-line programs. In: Proc. Computer Algebra in Scientific Computing
  (CASC 2013). Vol. 8136 of Lecture Notes in Computer Science. pp. 61--74.
\newline\urlprefix\url{http://dx.doi.org/10.1007/978-3-319-02297-0_5}

\bibitem[{Arnold et~al.(2014)Arnold, Giesbrecht, and Roche}]{AGR14}
Arnold, A., Giesbrecht, M., Roche, D.~S., 2014. Sparse interpolation over
  finite fields via low-order roots of unity. In: Proceedings of the 39th
  International Symposium on Symbolic and Algebraic Computation. ISSAC '14.
  ACM, New York, NY, USA, pp. 27--34.
\newline\urlprefix\url{http://dx.doi.org/10.1145/2608628.2608671}

\bibitem[{Arnold and Roche(2014)}]{AR14}
Arnold, A., Roche, D.~S., 2014. Multivariate sparse interpolation using
  randomized {K}ronecker substitutions. In: Proceedings of the 39th
  International Symposium on Symbolic and Algebraic Computation. ISSAC '14.
  ACM, New York, NY, USA, pp. 35--42.
\newline\urlprefix\url{http://dx.doi.org/10.1145/2608628.2608674}

\bibitem[{Ben-Or and Tiwari(1988)}]{BT88}
Ben-Or, M., Tiwari, P., 1988. A deterministic algorithm for sparse multivariate
  polynomial interpolation. In: Proceedings of the twentieth annual ACM
  symposium on Theory of computing. STOC '88. ACM, New York, NY, USA, pp.
  301--309.
\newline\urlprefix\url{http://dx.doi.org/10.1145/62212.62241}

\bibitem[{B\"{u}rgisser et~al.(1997)B\"{u}rgisser, Clausen, and
  Shokrollahi}]{clausen1997algebraic}
B\"{u}rgisser, P., Clausen, M., Shokrollahi, M.~A., 1997. Algebraic Complexity
  Theory. Vol. 315. Springer.

\bibitem[{Cand{\'e}s et~al.(2006)Cand{\'e}s, Romberg, and Tao}]{CRT06}
Cand{\'e}s, E.~J., Romberg, J.~K., Tao, T., 2006. Stable signal recovery from
  incomplete and inaccurate measurements. Communications on Pure and Applied
  Mathematics 59~(8), 1207--1223.
\newline\urlprefix\url{http://dx.doi.org/10.1002/cpa.20124}

\bibitem[{Cantor and Kaltofen(1991)}]{CanKal91}
Cantor, D.~G., Kaltofen, E., Oct. 1991. On fast multiplication of polynomials
  over arbitrary algebras. Acta Inf. 28~(7), 693--701.
\newline\urlprefix\url{http://dx.doi.org/10.1007/BF01178683}

\bibitem[{Dickson(1901)}]{Dix01}
Dickson, L.~E., 1901. Linear groups with an exposition of the Galois field
  theory / by Leonard Eugene Dickson. Leipzig :B.G. Teubner,,
  http://www.biodiversitylibrary.org/bibliography/22174.
\newline\urlprefix\url{http://www.biodiversitylibrary.org/item/62667}

\bibitem[{Donoho(2006)}]{Don06}
Donoho, D., April 2006. Compressed sensing. Information Theory, IEEE
  Transactions on 52~(4), 1289--1306.
\newline\urlprefix\url{http://dx.doi.org/10.1109/TIT.2006.871582}

\bibitem[{F\"{u}rer(2009)}]{Fur09}
F\"{u}rer, M., 2009. Faster integer multiplication. SIAM Journal on Computing
  39~(3), 979--1005.
\newline\urlprefix\url{http://dx.doi.org/10.1137/070711761}

\bibitem[{Garg and Schost(2009)}]{GarSch09}
Garg, S., Schost, E., Jun. 2009. Interpolation of polynomials given by
  straight-line programs. Theor. Comput. Sci. 410~(27-29), 2659--2662.
\newline\urlprefix\url{http://dx.doi.org/10.1016/j.tcs.2009.03.030}

\bibitem[{Gathen and Gerhard(2003)}]{MCA}
Gathen, J. V.~Z., Gerhard, J., 2003. Modern Computer Algebra, 2nd Edition.
  Cambridge University Press, New York, NY, USA.

\bibitem[{Giesbrecht et~al.(2009)Giesbrecht, Labahn, and Lee}]{GLL09}
Giesbrecht, M., Labahn, G., Lee, W., 2009. Symbolic-numeric sparse
  interpolation of multivariate polynomials. Journal of Symbolic Computation
  44~(8), 943 -- 959.
\newline\urlprefix\url{http://dx.doi.org/10.1016/j.jsc.2008.11.003}

\bibitem[{Giesbrecht and Roche(2011)}]{GieRoc11}
Giesbrecht, M., Roche, D.~S., 2011. Diversification improves interpolation.
  ISSAC '11, 123--130.
\newline\urlprefix\url{http://dx.doi.org/10.1145/1993886.1993909}

\bibitem[{Hassanieh et~al.(2012)Hassanieh, Indyk, Katabi, and
  Price}]{HasIndKatPri12}
Hassanieh, H., Indyk, P., Katabi, D., Price, E., 2012. Nearly optimal sparse
  fourier transform. In: Proceedings of the Forty-fourth Annual ACM Symposium
  on Theory of Computing. STOC '12. ACM, New York, NY, USA, pp. 563--578.
\newline\urlprefix\url{http://doi.acm.org/10.1145/2213977.2214029}

\bibitem[{Hoeffding(1963)}]{Hoe63}
Hoeffding, W., 1963. Probability inequalities for sums of bounded random
  variables. Journal of the American Statistical Association 58~(301), 13--30.
\newline\urlprefix\url{http://dx.doi.org/10.1080/01621459.1963.10500830}

\bibitem[{Huang and Rao(1999)}]{HR99}
Huang, M.-D.~A., Rao, A.~J., 1999. Interpolation of sparse multivariate
  polynomials over large finite fields with applications. Journal of Algorithms
  33~(2), 204--228.
\newline\urlprefix\url{http://dx.doi.org/10.1006/jagm.1999.1045}

\bibitem[{Javadi and Monagan(2010)}]{JM10}
Javadi, S. M.~M., Monagan, M., 2010. Parallel sparse polynomial interpolation
  over finite fields. In: Proceedings of the 4th International Workshop on
  Parallel and Symbolic Computation. PASCO '10. ACM, New York, NY, USA, pp.
  160--168.
\newline\urlprefix\url{http://dx.doi.org/10.1145/1837210.1837233}

\bibitem[{Kaltofen(1988)}]{Kal88:frag}
Kaltofen, E.~L., 1988. Unpublished article fragment.
\newline\urlprefix\url{http://www.math.ncsu.edu/~kaltofen/
  bibliography/88/Ka88\_ratint.pdf}

\bibitem[{Kaltofen(2010{\natexlab{a}})}]{Kal10:pasco}
Kaltofen, E.~L., 2010{\natexlab{a}}. Fifteen years after {DSC} and {WLSS2}
  {What} parallel computations {I} do today. In: Proc. International Workshop
  on Parallel Symbolic Computation (PASCO 2010). pp. 10--17.
\newline\urlprefix\url{http://dx.doi.org/10.1145/1837210.1837213}

\bibitem[{Kaltofen(2010{\natexlab{b}})}]{Kal10a}
Kaltofen, E.~L., 2010{\natexlab{b}}. Fifteen years after {DSC} and {WLSS2}:
  {W}hat parallel computations {I} do today [invited lecture at {PASCO} 2010].
  In: Proceedings of the 4th International Workshop on Parallel and Symbolic
  Computation. PASCO '10. ACM, New York, NY, USA, pp. 10--17.
\newline\urlprefix\url{http://dx.doi.org/10.1145/1837210.1837213}

\bibitem[{Kaltofen et~al.(2011)Kaltofen, Lee, and Yang}]{KLY11}
Kaltofen, E.~L., Lee, W.-s., Yang, Z., 2011. Fast estimates of hankel matrix
  condition numbers and numeric sparse interpolation. In: Proceedings of the
  2011 International Workshop on Symbolic-Numeric Computation. SNC '11. ACM,
  New York, NY, USA, pp. 130--136.
\newline\urlprefix\url{http://dx.doi.org/10.1145/2331684.2331704}

\bibitem[{Kushilevitz and Mansour(1993)}]{KusMan93}
Kushilevitz, E., Mansour, Y., 1993. Learning decision trees using the fourier
  spectrum. SIAM Journal on Computing 22~(6), 1331--1348.
\newline\urlprefix\url{http://dx.doi.org/10.1137/0222080}

\bibitem[{Le~Gall(2014)}]{LeG14}
Le~Gall, F., 2014. Powers of tensors and fast matrix multiplication. In:
  Proceedings of the 39th International Symposium on Symbolic and Algebraic
  Computation. ISSAC '14. ACM, New York, NY, USA, pp. 296--303.
\newline\urlprefix\url{http://doi.acm.org/10.1145/2608628.2608664}

\bibitem[{Pritchard(1982)}]{Pri82}
Pritchard, P., 1982. Explaining the wheel sieve. Acta Informatica 17~(4),
  477--485.
\newline\urlprefix\url{http://dx.doi.org/10.1007/BF00264164}

\bibitem[{Prony(1795)}]{Pro95}
Prony, Gaspard-Clair-Fran\c{c}ois-Marie~Riche, B.~d., 1795. {Essai
  exp\'erimental et analytique sur les lois de la Dilatabilit\'e des fluides
  \'elastique et sur celles de la Force expansive de la vapeur de l'eau et de
  la vapeur de l'alkool, \`a diff\'erentes temp\'eratures}. J. de l'\'Ecole
  Polytechnique 1, 24--76.

\bibitem[{Rosser and Schoenfeld(1962)}]{RosSch62}
Rosser, J.~B., Schoenfeld, L., 1962. Approximate formulas for some functions of
  prime numbers. Illinois J. Math. 6, 64--94.
\newline\urlprefix\url{http://projecteuclid.org/euclid.ijm/1255631807}

\bibitem[{Schwartz(1980)}]{Zip80}
Schwartz, J.~T., Oct. 1980. Fast probabilistic algorithms for verification of
  polynomial identities. J. ACM 27~(4), 701--717.
\newline\urlprefix\url{http://doi.acm.org/10.1145/322217.322225}

\bibitem[{Shoup(1994)}]{Sho94}
Shoup, V., 1994. Fast construction of irreducible polynomials over finite
  fields. Journal of Symbolic Computation 17~(5), 371--391.
\newline\urlprefix\url{http://dx.doi.org/10.1006/jsco.1994.1025}

\bibitem[{Zippel(1990)}]{Zip90}
Zippel, R., 1990. Interpolating polynomials from their values. Journal of
  Symbolic Computation 9~(3), 375--403, computational algebraic complexity
  editorial.
\newline\urlprefix\url{http://dx.doi.org/10.1016/S0747-7171(08)80018-1}

\end{thebibliography}



\appendix
\newpage
\section{Notation}\label{app:notation}

\centering\footnotesize
\begin{tabularx}{\textwidth}{ XllX }
\hline
\multicolumn{4}{ c }{ Mathematical Objects }\\
\hline
&$m,n,p,q,r,s,t,u \in \ZZ_{>0}$ & & \\
&$\epsilon, \mu \in (0,1)$ & probabilities& \\
&$\lambda \in \mathbb{R}_{>0}$ & constant& \\
&$\ZZ_p$ & ring of integers modulo $p$& \\
&$\FF_q$ & finite field of size $q$& \\
&$[n]$ & $\eqdef \{1,2,\dots, n\}$ & \\
&$[n_1, \dots, n_s]$ &$\eqdef \{(i_1, \dots, i_s) \mid i_k \in 
[n_k], k \in [s]\}$ & \\
&$(i,j,k,\ell) \in [m,n,s,t]$ & indices& \\
&$x,\z_j, j \in [n]$ & indeterminates & \\
&$\FF_q[\Z] \eqdef \FF_q[\z_1, \dots, \z_n]$ & polynomial ring 
with $n$ indeterminates& \\
&$\F \in \FF_q[\Z]$ & polynomial& \\
&$\SLP$ & straight-line program computing $\F \in \FF_q[\Z]$& \\
&$L$ & length of $\SLP$& \\
&$c, c_\ell \in \FF_q, \ell \in [t], b \in \FF_{q^u}$ & 
coefficients& \\
&$\bee \FF_{q^u}^{s+1}$ & vector of coefficients from a field 
extension& \\
&$\e, \e_\ell \in \ZZ^n$ & exponent vectors& \\
&$p_i, i \in [m]$ & randomly selected primes& \\
&$\ve, \ve_{ij} \in \ZZ_{p_i}^n, (i,j) \in [m,n]$ & randomly 
selected vectors& \\
&$V_i = [v_{ijk}] \in \ZZ_{p_i}^{n \times n}$ & matrix whose rows 
are $\ve_{ij}$, $j \in [n]$& \\
&$\avec, \avec_k \in \FF_{q^u}, k \in [s]$ & randomly selected 
vectors& \\\hline
\multicolumn{4}{c}{Polynomial Notation} \\
\hline
&$\e \bmod p$ & $\eqdef (e_1 \bmod p, \dots, e_n \bmod p)$ & \\
&$\avec^\e$ & $\eqdef a_1^{e_1}\cdots a_n^{e_n}$& \\
&$c\Z^e$ & $\eqdef cz_1^{e_1}\cdots z_t^{e_t}$& \\
&$\F(\Z)$ & $\eqdef \F(\z_1, \dots, \z_n)$& \\
&$\F(\avec\Z)$ & $\eqdef \F(a_1\z_1, \dots, a_n\z_n)$ & \\
&$\F(x^\ve)$ & $\eqdef \F(x^{v_1}, \dots, x^{v_n})$ & \\
&$\F(\avec x^\ve)$ &$\eqdef \F(a_1 x^{v_1}, \dots, 
a_n x^{v_n})$ & \\
&$\F \bmod (\Z^p-1)$ & $\eqdef \F \bmod (\z_1^p-1, \dots, 
\z_n^p-1)$. & \\
\end{tabularx}

\begin{tabularx}{\textwidth}{ XcX }
\hline
&Defined Constants& \\
\hline
\end{tabularx}\vspace{-0.5em}
\begin{align*}
m &= \mdef,\\
\lambda &= \lambdadef,\\
s &= \sdef,\\
u &= \udef.\\
\end{align*}

\end{document}